\documentclass[11pt,a4paper]{llncs}
\usepackage{amsmath}
\usepackage{amssymb}
\usepackage{graphicx}
\usepackage{marvosym}
\usepackage{url}
\usepackage{fancyhdr}
\usepackage{cite}
\usepackage{amsmath,amssymb,amsfonts}
\usepackage{textcomp}
\usepackage{lipsum,multicol}
\usepackage{tikz}
\usepackage{diagbox}
\usepackage{array,multirow,makecell}
\usepackage{listings}
\definecolor{darkWhite}{rgb}{0.94,0.94,0.94}
\usepackage{tabularx,booktabs}
\newcolumntype{C}{>{\centering\arraybackslash}X}
\usepackage{lscape}

\renewcommand\qed{$\blacksquare$}
\renewcommand{\textbf}[1]{\begingroup\bfseries\mathversion{bold}#1\endgroup} 
\usepackage{geometry}
\newcommand{\keywords}[1]{\par\addvspace\baselineskip
\noindent\keywordname\enspace\ignorespaces#1}

\fancypagestyle{firstpage}{\fancyhf{}
}
\begin{document}
\title{\LARGE{Fixed-Point Code Synthesis for Neural Networks\thanks{This work is supported by La  Region  Occitanie  under  Grant  GRAINE  -  SYFI. https://www.laregion.fr}}}
\author{\large{Hanane Benmaghnia\inst{1}  \and Matthieu Martel\inst{1,2} \and Yassamine Seladji \inst{3}}}
\institute{\large{University of Perpignan Via Domitia, Perpignan, France
  \and
Numalis, Cap Omega, Rond-point Benjamin Franklin 34960 Montpellier, France
\and
University of Tlemcen Aboubekr Belkaid, Tlemcen, Algeria\\
\inst{1}\email{\{first.last\}@univ-perp.fr}, \inst{3}\email{yassamine.seladji@univ-tlemcen.dz}}}
\maketitle
\thispagestyle{firstpage}
\begin{abstract}
Over the last few years, neural networks have started penetrating safety critical systems 
to take decisions in  robots, rockets, autonomous driving car, etc.
A problem is that these critical systems often have limited computing resources. Often, they use the fixed-point arithmetic for its many advantages (rapidity, compatibility with small memory devices.)
In this article, a new technique is introduced to tune the formats (precision) of already trained neural networks using fixed-point arithmetic, which can be implemented using integer operations only. The new optimized neural network computes the output with fixed-point numbers without modifying the accuracy up to a threshold fixed by the user. A fixed-point code is synthesized for the new optimized neural network ensuring the respect of the threshold for any input vector belonging the range $[x_{min},\, x_{max}]$ determined during the analysis. From a technical point of view, we do a preliminary analysis of our floating neural network to determine the worst cases, then we generate a system of linear constraints among integer variables that we can solve by linear programming. The solution of this system is the new fixed-point format of each neuron. The experimental results obtained show the efficiency of our method which can ensure that the new fixed-point neural network has the same behavior as the initial floating-point neural network.
\keywords{Computer Arithmetic, Code Synthesis, Formal Methods, Linear Programming, Numerical Accuracy, Static Analysis.}
\end{abstract}

\newcommand{\x}[1]
{
\hat{#1}
}

\section{Introduction}

Nowadays, neural networks have become increasingly popular. They have started penetrating safety critical domains and embedded systems, in which they are often taking important decisions such as autonomous driving cars, rockets, robots, etc.
These neural networks become larger and larger while embedded systems still have limited resources (memory, CPU, etc.)
As a consequence, using and running deep neural networks \cite{sun2018testing} on embedded systems with limited resources introduces several new challenges \cite{GSS19,JinDLTTC19,enderich2019fix,ChenHZX17,LinTA16, joseph2020programmable, LauterV20,HanZZWL19}.
The fixed-point arithmetic is more adapted for these embedded systems which often have a working processor with integers only.
The approach developed in this article concerns the fixed-point and integer arithmetic applied to trained neural networks (NNs).
NNs are trained on computers with a powerful computing unit using most of the time the IEEE754 floating-point arithmetic~\cite{IEEE,martel2017floating}. 
Exporting NNs using fixed-point arithmetic can perturb or change the answer of the NNs which are in general sensible to the computer arithmetic. 
A new approach is required to adapt NN computations to the simpler CPUs of embedded systems. This method consists in using fixed-point  arithmetic because it is faster and lighter to manipulate for a CPU while it is more complicated to handle for the developer.
We consider the problem of tuning the formats (precision) of an already trained floating-point NN, in such a way that, after tuning, the synthesized fixed-point NN behaves almost like the original performing computations. More precisely, if the NN is an interpolator, i.e. NNs computing mathematical functions, the original NN (floating-point) and the new NN (fixed-point) must behave identically if they calculate a given function $f$, such that, the absolute error (Equation~(\ref{er})) between the numerical results computed by both of them is equal to or less than a threshold set by the user. If the NN is a classifier, the new NN have to classify correctly the outputs in the right category comparing to the original NN. This method is developed in order to synthesize NNs fixed-point codes using integers only.
This article contains nine sections and an introductory example in Section~\ref{s2}, where we present our method in a simplified and intuitive way. Some notations are introduced in Section~\ref{s10}. In Section~\ref{s3}, we present the fixed-point arithmetic, where we show how to represent a fixed-point number and the elementary operations.
The errors of computations and conversions inside a NN are introduced in Section~\ref{s4}. Section~\ref{s5} deals with the generation of constraints to compute the optimal format for each neuron using linear programming \cite{welke2020ml2r}. Our tool and its features are presented in Section \ref{s9}. Finally, we demonstrate the experimental results in Section~\ref{s6} in terms of accuracy and bits saved. Section~\ref{s7} presents the related work then Section~\ref{s8} concludes and gives an overview of our future work.
\section{Notations}
\label{s10}
In the following sections, we will use these notations: \\
\textbf{$\bullet$}$\mathbb{\aleph}$: Set of fixed-point numbers. \textbf{$\bullet$}$\mathbb{N}$: Set of natural integers. \textbf{$\bullet$}$\mathbb{Z}$: Set of relative integers. \textbf{$\bullet$}$\mathbb{R}$:~Set of real numbers. \textbf{$\bullet$}$\mathbb{F}$: Set of IEEE754 floating-point numbers \cite{IEEE}. \textbf{$\bullet$NN}: Neural Network.
 \textbf{$\bullet<M^{\x{x}}$}, \textbf{$L^{\x{x}}>$}: Format of the fixed-point number $\x{x}$ where \textbf{$M^{\x{x}}$} represents the Most significant bit (integer part) and \textbf{$L^{\x{x}}$} the Least significant bit (fractional part). \textbf{$\bullet \epsilon_{{\x{x}}}$}: Error on the fixed-point number ${\x{x}}$. \textbf{$\bullet b$}: Bias. \textbf{$\bullet W$}: Matrix of weights. \textbf{$\bullet m$}:
  Number of layers of a neural network. \textbf{$\bullet n$}: Number of neurons by layer. \textbf{$\bullet k$}: Index of layer. \textbf{$\bullet i$}: Index of neuron. \textbf{$\bullet$ufp}: Unit in the first place \cite{martel2017floating,IEEE}. \textbf{$\bullet$} \textbf{ReLU}: Rectified Linear Unit~\cite{AI2, sharma2017activation,abs-2104-02466}. \textbf{$\bullet T$}: size of data types ($8$, $16$, $32$ bits.) \textbf{$\bullet \oplus$}: Fixed-point addition. \textbf{$\bullet \otimes$}: Fixed-point multiplication.
  \section{An Introductory Example}
\label{s2}
In this section, we present a short example of a fully connected neural network \cite{albawi2017understanding} containing three layers ($m=3$) and two neurons by layer ($n=2$) as shown in Figure~\ref{fig1}. The objective is to give an intuition of our approach.

Our main goal is to synthesize a fixed-point code for an input NN with an error threshold between $0$ and $1$ defined by the user, and respecting the initial NN which uses the floating-point arithmetic \cite{IEEE,martel2017floating}.
The error threshold is the maximal absolute error accepted between the original floating-point NN and the synthesized fixed-point NN in all the outputs of the output layer (max norm). This absolute error is computed by substracting the fixed-point value to the floating-point value (IEEE754 with single precision \cite{IEEE,martel2017floating}) as defined in Equation~\eqref{er}.
To compute this error, we convert the fixed-point value into a floating-point value.
\begin{equation}
\label{er}
%Absolute \, error= |second \, column - third \, column|
Absolute \, error= |FloatingPoint \, Result - FixedPoint \, Result|
\end{equation}
In this example, the threshold is $0.02$ and the data type $T=32$ bits. In other words, the resulting error of all neurons in the output layer ($u_{30}$, $u_{31}$ of Figure \ref{fig1}) must be equal to or less than $0.02$ using integers in $32$ bits.
\def\layersep{2.6cm}
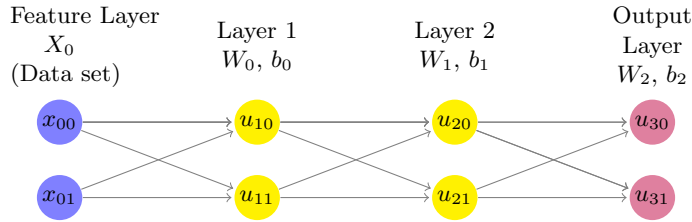
\begin{figure}[h]
\centering
\begin{center}
\begin{tikzpicture}[shorten >=0.95pt,->,draw=black!50, node distance=1*\layersep]

    \tikzstyle{every pin edge}=[<-,shorten <=0.95pt]
    \tikzstyle{neuron}=[circle,fill=black!25,minimum size=17pt,inner sep=0.0pt]
    \tikzstyle{input neuron}=[neuron, fill=blue!50];
    \tikzstyle{output neuron}=[neuron, fill=purple!50];
    \tikzstyle{hidden neuron}=[neuron, fill=yellow];
    \tikzstyle{annot} = [text width=4em, text centered]

% Draw the input layer nodes
\foreach \name / \y in {0,...,1}
    \node[input neuron] (I-\name) at (0,-\y) {$x_{0\y}$};
  
% Draw the hidden layer nodes
\foreach \name / \y in {0,...,1}
        \path[yshift=0cm]
            node[hidden neuron] (A-\name) at (\layersep,-\y cm) {$u_{1\y}$};

\foreach \name / \y in {0,...,1}
        \path[yshift=0cm]
            node[hidden neuron] (B-\name) at (2*\layersep,-\y cm) {$u_{2\y}$};
% Draw the output layer node
\node[output neuron] (O) at (\layersep+\layersep+\layersep, -0.0) {$u_{30}$};

\node[output neuron] (1) at (\layersep+\layersep+\layersep, -1.0) {$u_{31}$};

% Connect every node in the input layer with every node in the hidden layer.
    \foreach \source in {0,...,0}
    \foreach \dest in {0,...,0}
        \path (I-\source) edge node[above]{} (A-\dest);
        
    \foreach \source in {1,...,1}
    \foreach \dest in {1,...,1}
        \path (I-\source) edge node[above]{} (A-\dest);  
       
     \foreach \source in {0,...,1}
    \foreach \dest in {0,...,0}
        \path (I-\source) edge node[above]{} (A-\dest);
     \foreach \source in {0,...,0}
    \foreach \dest in {0,...,1}
        \path (I-\source) edge node[above]{} (A-\dest);
                    
% Connection de la hidden layer 1 à la hidden layer 2 :
        \foreach \source in {0,...,0}
        \foreach \dest in {0,...,0}
            \path (A-\source) edge node[above]{}(B-\dest);
            
        \foreach \source in {1,...,1}
        \foreach \dest in {1,...,1}
            \path (A-\source) edge node[above]{}(B-\dest);
         \foreach \source in {0,...,0}
        \foreach \dest in {0,...,1}
            \path (A-\source) edge node[above]{}(B-\dest);      
            
              \foreach \source in {0,...,1}
        \foreach \dest in {0,...,0}
            \path (A-\source) edge node[above]{}(B-\dest);   
% Connect every node in the hidden layer with the output layer
    \foreach \source in {0,...,0}
    \foreach \dest in {0,...,0}
        \path (B-\source) edge node[above]{ } (O);
    \foreach \source in {0,...,1}
    \foreach \dest in {0,...,0}
        \path (B-\source) edge node[above]{ } (O);
        \foreach \source in {0,...,0}
    \foreach \dest in {0,...,1}
        \path (B-\source) edge  (1);
                
      \foreach \source in {1,...,1}
    \foreach \dest in {1,...,1}
        \path (B-\source) edge node[above]{ } (1);
        
% Annotate the layers

\node[annot, above of=A-0, node distance=1cm] (hl1) {Layer 1\\ $W_0, \, b_0$};
\node[annot, left of=hl1] {Feature~Layer $X_0$ (Data~set)};
\node[annot, right of=hl1] (hl2) {Layer 2\\ $W_1, \, b_1$};
\node[annot, right of=hl2] {Output Layer \\ $W_2, \, b_2$};
\end{tikzpicture}
\end{center}
\caption{A fully connected NN with 3 layers and 2 neurons per layer.}
\label{fig1}
\end{figure}
\newcommand*{\Vecteur}[3]{% 
  \ensuremath{{#1}=\, 
    \begin{pmatrix} 
      #2\\ 
      #3 
    \end{pmatrix}}}

\newcommand*{\Matrice}[5]{% 
\ensuremath{{#1}=\,
 \begin{pmatrix} 
 #2 & #3 \\ 
 #4 & #5 
 \end{pmatrix} }}

Hereafter, we consider the feature layer $X_0$, which corresponds to the input vector.
The biases are $b_0, \, b_1$ and $b_2$. The matrices  of weights are $W_0, \, W_1$ and $W_2$, such that each bias (respectively matrix of weights) corresponds to one layer.\\
The affine function for the $k^{th}$ layer is defined in a standard way as 
 \begin{small}
\begin{equation}
\label{f}
\begin{array}[t]{lrcl}
f_{k,i} : & \mathbb{F}^n &\longrightarrow  \mathbb{F}  &\\
    & X_{k-1} & \longmapsto u_{k,i}&=f_{k,i}(X_{k-1}) =\displaystyle \sum_{j=0}^{n-1} (w_{k-1,i,j}\times   x_{k-1,j})+  b_{k-1,i},   
\end{array}
\end{equation}
 \end{small}
$ \forall 1 \leq k \leq m$, $\forall 0 \leq i < n$, where $X_{k-1}$= $(x_{k-1,0},..., \,x_{k-1,n-1})^t$ is the input vector (the input of the layer $k$ is the output of the layer $k-1$), $b_{k-1}$~=~ $(b_{k-1,0},..., \, b_{k-1,n-1})^t$~$\in~ \mathbb{F}^n$ and $W_{k-1} \in \mathbb{F}^{n \times n}$  
($w_{k-1,i,j}$ is the coefficient in the line $i$ and column $j$ in $W_{k-1}$.) \\
Informally, a fixed-point number is represented by an $integer \, value$ and a $format$ $<M,L>$ which gives us the information about the number $M$ of significant bits before the binary point and $L$ the number of significant bits after the binary point required for each coefficient in $W_{k-1}$, $b_{k-1}$, $X_0$, and the output of each neuron $u_{k,i}$. We notice that, at the beginning, we convert the input vector $X_0$ in the size of the data type required $T$.
In fixed-point arithmetic, before computing the affine function defined in Equation~(\ref{f}), we need to know the optimal formats $<M,L>$ of all the coefficients. 
To compute these formats, we generate automatically linear constraints according to the given threshold. These linear constraints formally defined in Section \ref{s5} are solved by linear programming \cite{welke2020ml2r}, and they give us the optimal value of the number of significant bits after the binary point $L$ for each neuron. We show in Equation~(\ref{sys}) some constraints generated for the neuron $u_{31}$ of the NN of Figure~\ref{fig1}.
  \begin{small}
\begin{equation}
\label{sys}
\begin{cases}
 L^{u}_{31} \geq 6 \qquad \qquad \, \, \, \,  \\
 L^{u}_{31} + M^{u}_{31} \leq 31 \\
 L^{u}_{31} \leq L^{x}_{20}, \qquad \, \, \, \, \, \, \\
  L^{u}_{31} \leq L^{x}_{21}, \qquad \quad \\
  ...
 \end{cases}
\end{equation}
  \end{small}
We notice that $L^{x}_{20}$ (respectively $L^{x}_{21}$) is the length of the fractional part of ${u}_{20}$ (respectively~${u}_{21}$.)
The first constraint gives a lower bound for $L^{u}_{31}$, so the output $u_{31}$ in the output layer has to fullfil the threshold fixed by the user and the error done must be equal to or less than this one. In other words, the number of significant bits of each neuron in the output layer must be equal at least to $6$ (if it is greater than $6$, this means that we are more accurate.)
The value $6$ is obtained by computing the unit in the first place \cite{IEEE,martel2017floating} of the threshold defined as
 \begin{small}
\begin{equation}
\label{ufp}
\forall x \in \mathbb{F},\quad \textnormal{ufp}(x)= \textnormal{min}\, \{ i \in \mathbb{N}:\, 2^{i+1}>x\}=\lfloor \textnormal{log}_2(x) \rfloor
\end{equation}
 \end{small}
The second constraint avoids overflow and ensures compliance to the data type chosen by the user (integers on $8$, $16$ or $32$ bits.) The third and fourth constraints ensure that the length of the fractional part $L_{31}^{u}$ computed is less than or equal to the length of the fractional parts of all its inputs  ($L_{20}^{x}$ and  $L_{21}^{x}$.)
Using the formats resulting from the solver, firstly, we convert all the coefficients of weights matrices $W_{k-1}$, biases $b_{k-1}$ and inputs $X_0$ from floating-point values to fixed-point values. 

Let  $b_{k-1}, \, W_{k-1}$ and $ X_0$ used in this example be  
\begin{center}
$\Vecteur{b_0}{-2}{4.5}$, $\Vecteur{b_1}{1.2}{0.5}$, $\Vecteur{b_2}{3}{1}$, $\Matrice{W_0}{3.5}{0.25}{-1.06}{4.1}$, $\Matrice{W_1}{-0.75}{4.85}{2.1}{0.48}$, $\Matrice{W_2}{-5}{12.4}{0.2}{-2}$, $\Vecteur{X_0}{2}{0.5}$.\\
\end{center}
Table \ref{t} presents the output results for each neuron. The floating-point results are shown in the second column and the fixed-point results in the third one. The last column contains the absolute error defined in Equation~\eqref{er} for the output layer ($u_{30}$, $u_{31}$) only.
\newcommand*{\Rouge}[1]{% 
\textcolor{red}{#1}}
\newcommand*{\plus}{% 
\oplus}
\newcommand*{\mult}{% 
\otimes}
\begin{table}[!]
\begin{center}
\caption{\centering Comparison between the floating-point and the fixed-point results corresponding to the NN of Figure~\ref{fig1}.}
\label{t}

\begin{tabular}[c]{|p{1.5cm}|p{3.6cm}|p{4cm}|c|}

\hline \centering \rotatebox{0}{\textbf{ Neuron}} & \centering \rotatebox{0}{\textbf{ Floating-Point Result}} & \centering \rotatebox{0}{\textbf{Fixed-Point Result}} 
& \rotatebox{0}{\textbf{Absolute Error}}  \\
\hline \centering \textbf{$u_{10}$} & \centering 5.125 & \centering 2624$<$3,9$>$= 5.125 & / \\
\hline \centering \textbf{$u_{11}$} & \centering 4.43 & \centering 4535$<$2,10$>$= 4.4287 & /\\
\hline \centering \textbf{$u_{20}$} & \centering 18.8417 & \centering 9643$<$5,9$>$= 18.8339 & /\\
\hline \centering \textbf{$u_{21}$} & \centering 13.3889 & \centering 6854$<$5,9$>$= 13.3867 & /\\
\hline \centering \textbf{$u_{30}$} & \centering 74.8136 & \centering 76620$<$9,10$>$= 74.8247 & \textcolor{red}{$1.06 \times 10^{-2}$ \textbf{$\approx 2^{-7}$}}\\
\hline \centering \textbf{$u_{31}$} & \centering -22.0094 & \centering -22536 $<$7,10$>$= -22.0078 &  \textcolor{red}{$1.63 \times 10^{-3}$ \textbf{$\approx {2^{-10}}$}}\\
\hline
\end{tabular}
\end{center}
\end{table}

 The error threshold fixed by the user at the beginning was 0.02 (6 significant bits after the binary point.)
 As we can see, the absolute error of the output layer in the Table \ref{t} is under the value of the threshold required. This threshold is fulfilled with our method using fixed-point arithmetic. 
Now, we can synthesize a fixed-point code for this NN respecting the user's threshold, the data type $T$, and ensuring the same behavior and quality as the initial floating-point NN.
\begin{small}
\begin{figure}[!h]
\begin{scriptsize}
\centering
	\lstset{numbers=left,numberstyle=\tiny\bfseries,stepnumber=1,firstnumber=1,numberfirstline=true}
	\begin{lstlisting}[language=C++]
	int main()
{	/* That NN has 3 layers and 2 neurons per layer */
 	int mul, u[3][2], x[4][2];	 	 	  

	x[0][0]=1073741824;	 	 	 // <1,29>
	x[0][1]=1073741824;	 	 	 // <-1,31>
	x[0][0]=x[0][0]>>6;	 	 	 // <1,23>
	mul=112*x[0][0];	 	 	    // <3,9>=<1,5>*<1,23>
	mul=mul>>19;
	u[1][0]=mul;	 	 	 	      // <3,9>
	x[0][1]=x[0][1]>>7;	 	 	 // <-1,24>
	mul=16*x[0][1];	 	 		    // <-2,14>=<-2,9>*<-1,24>
	mul=mul>>19;
	mul=mul>>5;
	u[1][0]=u[1][0]+mul;			  //<3,9>=<3,9>+<-2,9>
	u[1][0]=u[1][0]+-1024;			//<3,9>=<3,9>+<1,9>
	u[1][0]=max(0,u[1][0]);		// ReLU(u[1][0])
	x[1][0]=u[1][0];
	x[0][0]=x[0][0]>>3;	 	   // <1,20>
	mul=-543*x[0][0];	 	 	 	 // <2,10>=<0,9>*<1,20>
	mul=mul>>19;
	u[1][1]=mul;	 	 	 	 	    // <2,10>
	x[0][1]=x[0][1]>>5;	 	 	 // <-1,19>
	mul=262*x[0][1];	 	 	 	  // <2,10>=<2,10>*<-1,19>
	mul=mul>>19;
	u[1][1]=u[1][1]+mul;			  //<2,10>=<2,10>+<2,10>
	u[1][1]=u[1][1]+4608;			 //<2,10>=<2,10>+<2,10>
	u[1][1]=max(0,u[1][1]);		// ReLU(u[1][1])
	...
	return 0; }
	\end{lstlisting}
\end{scriptsize}
\caption{\centering Fixed-point code synthesized for the neurons $u_{10}$ and $u_{11}$ of Figure~\ref{fig1} on $32$~bits.}
\label{fig_code}
\end{figure}
\end{small}

Figure~\ref{fig_code} shows some lines of code synthesized by our tool for the neurons $u_{10}$ and $u_{11}$ using Equation \eqref{f} and the fixed-point arithmetic. The running code gives the results shown in the third column of the Table \ref{t}. For example, the line $5$ represents the input $x_{00}=2$ in the fixed-point representation. This value is shifted on the right through $6$~bits (line $7$) in order to be aligned and used in the multiplication (line $8$) by $w_{000}=\,3.5$ represented by $112$ in the fixed-point arithmetic.
The fixed-point output $u_{10}$ ($2624$) in the Table~\ref{t} is returned by the line $16$. 
\newcommand*{\Bleu}[1]{% 
\textcolor{blue}{#1}}
\vspace{0.03cm}
\newcommand*{\Olive}[1]{% 
\textcolor{OliveGreen}{#1}}
\section{\textbf{Fixed-Point Arithmetic}} 
\label{s3}
In this section, we briefly describe the fixed-point arithmetic as implemented in most digital computers \cite{catrina2010secure,yates2009fixed,bevcvavr2005fixed}. 
Since fixed-point operations rely on integer operations, computing with fixed-point numbers is highly efficient. 
We start by defining the representation of a fixed-point number in Subsection~\ref{4-a}, then we present briefly the operations needed (addition, multiplication and activation functions) in this article in Subsection~\ref{4-b}. 
\subsection{\textbf{Representation of a Fixed-Point Number}}
\label{4-a}
A fixed-point number is represented by an \textbf{integer} value and a \textbf{format} $<M,L>$ where $M \in \mathbb{Z}$ is the number of significant bits before the binary point and $L \in \mathbb{N}$ is the number of significant bits after the binary point. We write the fixed-point number $\x{a}= value_{<M^{\x{a}},\,L^{\x{a}}>}$ and define it in Definition \ref{def1}.
    \begin{definition}
    \label{def1}
Let us consider $\x{a} \in \aleph$, $A_{\hat{a}} \in \mathbb{N}$ such that \\$\x{a} = (-1) ^{s_{\hat{a}}}.A_{\hat{a}}.\beta^{-L_{\hat{a}}}$ and 
% $A_{\hat{a}} = (a_{n-1}a_{n-2}...a_1a_0)_\beta$, 
$P_{\hat{a}} = M_{\hat{a}} + L_{\hat{a}}+1$, where
$\beta$ is the basis of representation,
\textcolor{black}{$\x{a}$}~is the fixed-point number with implicit scale factor $\beta ^{-L_{\hat{a}}}$ (Figure~\ref{representation}), \textcolor{black}{$A_{\hat{a}}$} is the integer representation of $\x{a}$ in the basis $\beta$, 
\textcolor{black}{$P_{\hat{a}}\in \mathbb{N}$}, $P_{\hat{a}} = M_{\hat{a}} + L_{\hat{a}}+1$ is the length of $\x{a}$ and $s_{\hat{a}}\in \{0, \, 1\}$ is its sign.
\end{definition}
\begin{figure}[!h]
\centering
\includegraphics[scale=0.53]{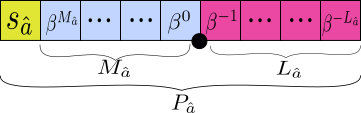}
\caption{Fixed-point representation of $\hat{a}$ in a format $<M^{\x{a}},L^{\x{a}}>$.}
\label{representation}
\end{figure}
The difficulty of the fixed-point representation is managing the position of the binary point manually against the floating-point representation which manages it automatically.
\begin{example}: The fixed-point value $ 3 <1,1>$ corresponds to $1.5$ in the floating-point representation. We have first to write $3$ in binary then we put the binary point at the right place (given by the format) and finally we convert it again into the decimal representation: $3_{10} <1,1> \,=\, 11_{2} <1,1> \,=\, 1.1_{2} = 1.5_{10}$. \end{example}
\subsection{\textbf{Elementary Operations}}
\label{4-b}
This subsection defines the elementary operations needed in this article like addition and multiplication which are used later in Equation~(\ref{e8}). We also define $\x{ReLU}$ (respectively $\x{Linear}$) in fixed-point arithmetic which corresponds to the activation function in some NNs \cite{AI2, sharma2017activation, abs-2104-02466}.
\begin{enumerate}
\item \textbf{Fixed-Point Addition}

Let us consider the two fixed-point numbers $\x{a}$,\, $\x{b}\in \aleph$ and their formats $<~M^{\x{a}},L^{\x{a}}>$, $<~M^{\x{b}},~L^{\x{b}}>$ respectively. Let $\plus$ be the fixed-point addition given by $\x{c} \in \aleph$, $\x{c}= \x{a}\plus \x{b}$.
\begin{figure}[h]
   \begin{minipage}[c]{.4\linewidth}
      \includegraphics[width=7cm]{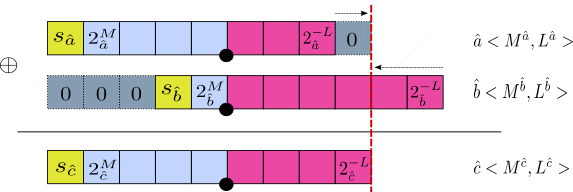}
      \caption{\centering Addition of two fixed-point numbers without a carry.}
\label{add}
   \end{minipage} \hfill
   \begin{minipage}[c]{.46\linewidth}
      \includegraphics[width=7cm]{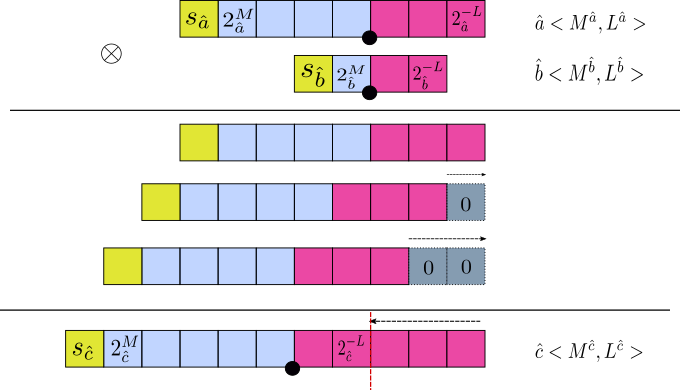}
      \caption{\centering Multiplication of two fixed-point numbers.}
\label{mult}
   \end{minipage}
\end{figure}

Figure~\ref{add} shows the fixed-point addition between $\x{a}$ and $\x{b}$. The fixed-point format required is $<M^{\x{c}},\,L^{\x{c}}>$. The objective is to have a result in this format, this is why we start by aligning the length of the fractional parts $L^{\x{a}}$ and $L^{\x{b}}$ according to $L^{\x{c}}$. If $L^{\x{a}}>L^{\x{c}}$, we truncate with $L^{\x{a}}-L^{\x{c}}$ bits, otherwise, we add $L^{\x{c}}-L^{\x{a}}$ zeros on the right hand of $\x{a}$, then   we do the same for $\x{b}$. The length of the integer part $M^{\x{c}}$ must be the maximum value between $M^{\x{a}}$ and $M^{\x{b}}$. If there is a carry, we add $+1$ to the number of bits of the integer part, otherwise the result is wrong. The algorithm of the fixed-point addition is given in \cite{lopez2014implementation,najahi2014synthesis}.
\item \textbf{Fixed-Point Multiplication}

Let us consider the two fixed-point numbers $\x{a}$,\, $\x{b}\in \aleph$ and their formats $<M^{\x{a}},\,L^{\x{a}}>$~, $<M^{\x{b}},\,L^{\x{b}}>$ respectively. Let $\mult$ be the fixed-point multiplication given by $\x{c} \in \aleph$, $\x{c}= \x{a}\mult \x{b}$.

Figure~\ref{mult} shows the fixed-point multiplication between $\x{a}$ and $\x{b}$. The fixed-point format required is $<M^{\x{c}},\,L^{\x{c}}>$. The objective is to have a result in this format, this is why we start by doing a standard multiplication which is composed by shifts and additions. 
If $(L^{\x{a}}+L^{\x{b}})>L^{\x{c}}$, we truncate with $(L^{\x{a}}+L^{\x{b}})-L^{\x{c}}$ bits, otherwise, we add $L^{\x{c}}-(L^{\x{a}}+L^{\x{b}})$ zeros on the right hand of $\x{c}$. The length of the integer part $M^{\x{c}}$ must be the sum of $M^{\x{a}}$ and $M^{\x{b}}$, otherwise the result is wrong. The algorithm of the fixed-point multiplication is given in \cite{lopez2014implementation,najahi2014synthesis}.
\item \textbf{Fixed-Point $\x{ReLU}$ }

Definition \ref{def2} defines the fixed-point $\x{ReLU}$ which is a non-linear activation function computing the positive values.
\begin{definition}
\label{def2}
Let us consider the fixed-point number $\x{a}=\, V_{\x{a}} \,<M^{\x{a}},L^{\x{a}}>$ $\in \aleph$ and the fixed-point zero written $\x{0}=\,$ $0\, <0,0>$  $\in \aleph$. Let $\x{c}=\, V_{\x{c}} \,<M^{\x{c}},L^{\x{c}}>$ $\in \aleph$ be the result of the fixed-point $\x{ReLU}$ given by 
\begin{small}
\begin{equation}
\label{relu}
\x{c}=\x{ReLU}(\x{a})=\, \x{max}(\x{0},\x{a}), 
\end{equation}
\end{small}
where \begin{small}$V_{\x{c}}=\, max(0,\, V_{\x{a}})$~and $<~M^{\x{c}},~ L^{\x{c}}>=$~$\begin{cases}
  <M^{\x{a}},L^{\x{a}}> \qquad if \quad V_{\x{c}}=\, V_{\x{a}},\\
 <0, \, 0> \quad \, \, \qquad otherwise. \\
 \end{cases}$\end{small}
\end{definition}
\item \textbf{Fixed-Point $\x{Linear}$ }

Definition \ref{def3} defines the fixed-point $\hat{Linear}$ which is an activation function returning the identity value.
\begin{definition}
\label{def3}
Let us consider the fixed-point number $\x{a}$ $\in \aleph$ with the format $<M^{\x{a}},\,L^{\x{a}}>$. Let $\x{c} \in \aleph$ be the result of the fixed-point $\x{Linear}$ activation function given by 
\begin{small}
\begin{equation}
\label{linear}
\x{c}=\x{Linear}(\x{a})=\x{a}.
\end{equation}
\end{small}
\end{definition}
\end{enumerate}
\section{\textbf{Error Modelling}}
\label{s4}
In this section, we introduce some theoretical results concerning the fixed-point arithmetic errors in Subsection \ref{s4_a} and we show the numerical errors done inside a NN in Subsection \ref{s4_b}. 
The error on the output of the fixed-point affine transformation function
 can be decomposed into two parts: the propagation of the input error and the computational error.
Hereafter, $\x{x}\in \aleph$ is used for the fixed-point representation with the format $<M^{\x{x}}, \, L^{\x{x}}>$ and $x\in \mathbb{F}$ for the floating-point representation. $\bar{X}\in \mathbb{F}^n$ is a vector of $n$ floating-point numbers and $\x{X}\in \aleph^n$ a vector of $n$ fixed-point numbers. 
\subsection{\textbf{Fixed-Point Arithmetic Error}}
\label{s4_a}
This subsection defines two important properties about errors made in fixed-point addition and multiplication which are used to compute affine transformations in a NN (substraction and division are useless in our context.) We start by introducing the propositions and then the proofs.
Proposition~\ref{p1} defines the error of the fixed-point addition  when we add two fixed-point numbers and Proposition~\ref{p2} defines the error due to the multiplication of two fixed-point numbers.
\begin{proposition}
\label{p1}
Let $\x{x}, \, \x{y}, \, \x{z}$ $\in$ $\aleph$ with a format $<M^{\x{x}},L^{\x{x}}>$ (respectively $<M^{\hat{y}},L^{\hat{y}}>$, $<M^{\hat{z}},L^{\hat{z}}>$.) Let $x, \, y, \, z\, \in \mathbb{F}$ be the floating-point representation of $\x{x}, \, \x{y}, \, \x{z}$.
Let $\epsilon_{\plus}\in \mathbb{R}$ be the error between the fixed-point addition $\x{z}=\x{x}\oplus\x{y}$ and the floating-point addition $z=x+y$. We have that
\begin{small}
\begin{equation}
\epsilon_{\plus}\leq 2^{-L^{\x{x}}}+2^{-L^{\x{y}}}+2^{-L^{{\x{z}}}}.
\end{equation}
\end{small}
\begin{proof}
Let us consider $\epsilon_{\hat{x}},$ $\epsilon_{\x{y}},$ $ \epsilon_{\x{z}}$ $\in \mathbb{R}$ errors of truncation in the fixed-point representation of $\x{x}$, $\x{y}$ and $\x{z}$ respectively. These ones are bounded by $2^{-L^{\x{x}}}$ ($2^{-L^{\x{y}}}$, $2^{-L^{\x{z}}}$ respectively) because $L^{\x{x}}$ ($L^{\x{y}}$, $L^{\x{z}}$ respectively) is the last correct bit in the fixed-point representation of $\hat{x}$ (respectively $\hat{y}$, $\hat{z}$.)\\
We have that \begin{small} $z=x+y$, $\x{z} = \x{x}\plus\x{y}$ \end{small} and
\begin{small}
$\epsilon_{\plus}\leq\epsilon_{\x{x}}+\epsilon_{\x{y}}+\epsilon_{\x{z}}$.
\end{small}
Then we obtain \begin{small}
$\epsilon_{\x{z}} \leq 2^{-L^{\x{x}}}+2^{-L^{\x{y}}}+2^{-L^{{\x{z}}}}$.~~\qed %\textsc{QED}
\end{small} 
\end{proof}
\end{proposition}

\begin{proposition}
\label{p2}
Let $\x{x}, \, \x{y}, \, \x{z}$ $\in$ $\aleph$ with a format $<M^{\x{x}},L^{\x{x}}>$ (respectively $<M^{\hat{y}},L^{\hat{y}}>	,<M^{\hat{z}},L^{\hat{z}}>$.) Let $x, \, y, \, z\, \in \mathbb{F}$ be the floating-point representation of $\x{x}, \, \x{y}, \, \x{z}$ in such a way $x=\, \x{x} \,+ \epsilon_{\x{x}}$ (respectively $y=\, \x{y} \,+ \epsilon_{\x{y}}$, $z=\, \x{z} \,+ \epsilon_{\x{z}}$.) \\
%where $M^{\hat{x}}, \, L^{\hat{x}}$ $\in$ $\mathbb{Z}$.  
Let $\epsilon_{\mult}\in \mathbb{R}$ be the resulting error between the fixed-point multiplication $\x{z}=\x{x}\otimes\x{y}$ and the floating-point multiplication $z=x\times y$. We have that
\begin{small}
\begin{equation}
\label{e7}
\epsilon_{\mult} \leq \x{y} \times 2^{-L^{\x{x}}}+\x{x} \times 2^{-L^{\x{y}}}+2^{-L^{\x{z}}}.
\end{equation}
\end{small}
\end{proposition}
\begin{proof}
Let us consider $\epsilon_{\hat{x}},$ $\epsilon_{\x{y}},$ $ \epsilon_{\x{z}}$ $\in \mathbb{R}$ errors of truncation of $\x{x}$, $\x{y}$ and $\x{z}$ respectively. These ones are bounded by $2^{-L^{\x{x}}}$ ($2^{-L^{\x{y}}}$, $2^{-L^{\x{z}}}$ respectively) because $L^{\x{x}}$ ($L^{\x{y}}$, $L^{\x{z}}$ respectively) is the last correct bit in the fixed-point representation of $\hat{x}$ (respectively $\hat{y}$, $\hat{z}$.)\\
We have that $z=x\times y$, $\x{z}=\x{x} \mult \x{y}$. We compute ($\x{x} \,+ \epsilon_{\x{x}}$) $\times$ ($\x{y} \,+ \epsilon_{\x{y}}$) and then we obtain 
\begin{small}
$\epsilon_{\mult}\leq{\x{y}} \times \epsilon_{\x{x}}+{\x{x}} \times \epsilon_{\x{y}}+\epsilon_{\x{x}} \times \epsilon_{\x{y}}+\epsilon_{\x{z}}$.
\end{small}
We get rid of the second order error $\epsilon_{\x{x}} \times \epsilon_{\x{y}}$ which is negligible in practice because our method needs to know only the most significant bit of the error which will be used in Equation~(\ref{c9}) in Section \ref{s5}. Now, the error becomes
\begin{small}
\begin{equation}
\label{e4}
\epsilon_{\mult}\leq{\x{y}} \times \epsilon_{\x{x}}+{\x{x}} \times \epsilon_{\x{y}}+\epsilon_{\x{z}}.
\end{equation}
\end{small}
Finally, we obtain
\begin{small}
$\epsilon_{\mult}\leq \x{y} \times 2^{-L^{\x{x}}}+\x{x} \times 2^{-L^{\x{y}}}+2^{-L^{\x{z}}}$. \hfill \qed		%\textsc{QED} 
\end{small}
\end{proof}

\subsection{\textbf{Neural Network Error}}
\label{s4_b}
Theoretical results about numerical errors inside a fully connected NN using fixed-point arithmetic are shown in this subsection. There are two types of errors: round off errors due to the computation of the affine function in Equation~\eqref{e8} and the propagation of the error of the input vector.

In a NN with fully connected layers \cite{albawi2017understanding},$\forall \, \bar{b} \in \mathbb{F}^n$, $\forall\, W \in \mathbb{F}^{n\times m}$, an output vector $\bar{u} \in \mathbb{F}^n$ is defined~as
\begin{small}
\begin{align}
\label{fx}
\begin{array}[t]{lrcl}
f : & \hspace{0.5cm} \mathbb{F}^m & \longrightarrow & \mathbb{F}^n \\
    &\bar{X} & \longmapsto & \bar{u}=f(\bar{X})=W. \bar{X}+\bar{b} \end{array}
\end{align}
\end{small}
 Proposition~\ref{p3} shows how to bound the numerical errors of Equation~(\ref{fx}) using fixed-point arithmetic.
\begin{figure}[!h]
\begin{center}
\includegraphics[scale=0.52]{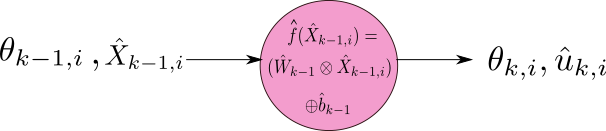}
\end{center}
\caption{Representation of $\hat{u}_{k,i}$ the $i^{th}$ neuron of the $k^{th}$ layer.}
\label{fig2}
\end{figure}
\begin{proposition}
\label{p3}
Let us consider the affine transformation as defined in Equation~(\ref{e8}). It represents the fixed-point version of Equation~\eqref{fx}. This transformation corresponds to what is computed inside a neuron (see Figure~\ref{fig2}.)
Let $\hat{u}_{k,i} \in \aleph$ be the fixed-point representation of $u_{k,i}$ and $\theta_{k,i}\in \mathbb{R}$ the error due to computations and conversions of the floating-point coefficients to fixed-point coefficients such that
\begin{small}
\begin{align}
\label{e8}
\begin{array}[t]{lrcl}
\x{f} : & \mathbb{\aleph}^n & \longrightarrow & \mathbb{\aleph} \\
    & \x{X}_{k-1} & \longmapsto & \x{u}_{k,i}=\hat{f}(\x{X}_{k-1}). \end{array}
\end{align}
where $\x{f}(\x{X}_{k-1})=\displaystyle\sum^{n-1}_{j=0} (\hat{w}_{k-1,i,j}\mult \hat{x}_{k-1,j})\plus \hat{b}_{k-1,i}$ and $\hat{X}_{k-1}=(\x{x}_{k-1,0},..., \, \x{x}_{k-1,n-1})^t$, $\, 1~\leq~k~\leq~m, \,~0 ~\leq ~i ~< ~n$.\\
\end{small}
Then the resulting error $\theta_{k,i}$  for each neuron $\hat{u}_{k,i}$ of each layer is given~by
\begin{small}
\begin{center}
\begin{equation}
\label{e10}
1 \leq k < m+1,\hspace{0.2cm} 0 \leq i < n, \quad\theta_{k,i} \leq \sum_{j=0}^{n-1}(2^{M^{\hat{x}}_{k-1,j}-L^{\hat{w}}_{k-1,i,j}}+2^{M^{\hat{w}}_{k-1,i,j}-L^{\hat{x}}_{k-1,j}}) + n\times2^{-L^{\hat{u}}_{k,i}}+2^{-L^{\hat{u}}_{k,i}+1}. 
\end{equation}
\end{center}
\end{small}
\end{proposition}
\begin{proof}
The objective is to bound the resulting error for each neuron (Figure~\ref{fig2}) of each layer due to affine transformations by bounding the error of the formula of Equation~(\ref{e8}).\\
We compute first the error of multiplication $\x{w}_{k-1,i,j}\mult~\x{x}_{k-1,j}$ using  Proposition~\ref{p2}. Then we bound the fixed-point sum of multiplications~$\displaystyle\sum^{n-1}_{j=0} (\hat{w}_{k-1,i,j}\mult \hat{x}_{k-1,j})$~and finally we use Proposition~\ref{p1} to bound the error of addition ~of $\displaystyle\sum^{n-1}_{j=0} (\hat{w}_{k-1,i,j}\mult~\hat{x}_{k-1,j}) \plus \hat{b}_{k-1,i}$.\\
Let $\epsilon_\alpha\in \mathbb{R}$ be the  error of $\hat{w}_{k-1,i,j}\mult \hat{x}_{k-1,j}$ and $2^{-L^{\hat{u}}_{k,i}}$ the truncation error of the output neuron $\hat{u}_{k,i}$. Using Proposition~\ref{p2}, we obtain
\begin{small}
 \begin{equation}
\label{e13}
 \epsilon_\alpha\leq 2^{M^{\hat{x}}_{k-1,j}-L^{\hat{w}}_{k-1,i,j}}+2^{M^{\hat{w}}_{k-1,i,j}-L^{\hat{x}}_{k-1,j}} +  2^{-L^{\hat{u}}_{k,i}}.
\end{equation}
\end{small}
Now, let us consider $\epsilon_\beta \in \mathbb{R}$ as the error of $\displaystyle\sum_{j=0}^{n-1}~(\hat{w}_{k-1,i,j}\mult~\hat{x}_{k-1,j})$. This error is computed by using the result of Equation~(\ref{e13}) such that
\begin{small}
\begin{equation}
\label{e14}
  \epsilon_\beta \leq \displaystyle \sum_{j=0}^{n-1}( 2^{-L^{\hat{w}}_{k-1,i,j}}\times 2^{M^{\hat{x}}_{k-1,j}} +2^{-L^{\hat{x}}_{k-1,j}}\times 2^{M^{\hat{w}}_{k-1,i,j}}) + \displaystyle \sum_{j=0}^{n-2}2^{-L^{\hat{u}}_{k,i}}+2^{-L^{\hat{u}}_{k,i}}. 
\end{equation}
\end{small}
Consequently,
\begin{small}
 \begin{equation}
 \label{e15}
 \epsilon_\beta \leq \displaystyle \sum_{j=0}^{n-1}(2^{M^{\hat{x}}_{k-1,j}-L^{\hat{w}}_{k-1,i,j}}+2^{M^{\hat{w}}_{k-1,i,j}-L^{\hat{x}}_{k-1,j}}) + n\times2^{-L^{\hat{u}}_{k,i}}.
 \end{equation}
 \end{small}
Finally, let $\epsilon_\gamma\in \mathbb{R}$ be  the error of $\displaystyle\sum_{j=0}^{n-1}~(\hat{w}_{k-1,i,j}\mult~\hat{x}_{k-1,j})\plus~\hat{b}_{k-1,i}$. Using Equation~(\ref{e15}) and Proposition~\ref{p1} we obtain $\epsilon_\gamma \leq$~$\epsilon_\beta +2^{-L^{\hat{u}}_{k,i}}$.
Finally,
\begin{small}
  \begin{equation}
   \label{e18}
 \epsilon_\gamma \leq \displaystyle \sum_{j=0}^{n-1}(2^{M^{\hat{x}}_{k-1,j}-L^{\hat{w}}_{k-1,i,j}}+2^{M^{\hat{w}}_{k-1,i,j}-L^{\hat{x}}_{k-1,j}})+ n\times2^{-L^{\hat{u}}_{k,i}}+2^{-L^{\hat{u}}_{k,i}+1}.
  \end{equation}
  \end{small}
If we combine Equations~\eqref{e10} and (\ref{e18}), we obtain
\begin{small}
\begin{equation}
\label{e24}
\theta_{k,i}=\epsilon_\gamma\leq\displaystyle \sum_{j=0}^{n-1}(2^{M^{\hat{x}}_{k-1,j}-L^{\hat{w}}_{k-1,i,j}}+2^{M^{\hat{w}}_{k-1,i,j}-L^{\hat{x}}_{k-1,j}})+ n\times2^{-L^{\hat{u}}_{k,i}}+2^{-L^{\hat{u}}_{k,i}+1}.\hspace{0.7cm} \textnormal{\qed}
\end{equation}
\end{small}
\end{proof}

In this section, we have bounded the affine transformation error $\theta_{k,i}$ for each neuron $\hat{u}_{k,i}$ of each layer $k$ of the NN in Equation~(\ref{e24}), respecting the equivalent floating-point computations. This resulting error $\theta_{k,i}$ is used in Section \ref{s5} to compute the optimal format $<M^{\x{u}}_{k,i},L^{\x{u}}_{k,i}>$ for each neuron $\hat{u}_{k,i}$.
\section{\textbf{Constraints Generation}}
In this section, we demonstrate how to generate the linear constraints automatically for a given NN, in order to optimize the number of significant bits after the binary point $L^{\x{u}}_{k,i}$ of the format $<M^{\x{u}}_{k,i},L^{\x{u}}_{k,i}>$ corresponding to the output $\x{u}_{k,i}$.
Let us remember that we have a floating-point NN with $m$ layers and $n$ neurons per layer working at some precision, and we want to compute a fixed-point NN with the same behavior than the initial floating-point NN for a given input vector. This new fixed-point NN must respect the threshold error and the data type $T$ $\in \{8,\, 16,\, 32\}$ bits for the C synthesized code.
The variables of the system of constraints are $L^{\x{u}}_{k,i}$ and $L^{\x{w}}_{k-1,i,j}$. They correspond respectively to the length of the fractional part of the output $\x{u}_{k,i}$ and $\x{w}_{k-1,i,j}$. We have $M^{\x{u}}_{k,i},\, M^{\x{x}}_{k-1,i},\, M^{\x{w}}_{k-1,i,j}\in~\mathbb{Z},$ and $L^{\x{u}}_{k,i},\, L^{\x{x}}_{k-1,i},\, L^{\x{w}}_{k-1,i,j}\in \mathbb{N}$, for $ 1 \leq k < m+1 $, and $0< i,\, j < n$, such that $M^{\x{u}}_{k,i}$ (respectively $M^{\x{x}}_{k-1,i},\, M^{\x{w}}_{k-1,i,j}$) can be negative when the value of the floating-point number is between $-1$ and $1$.
We have also, $M^{\x{w}}_{k-1,i,j}$ (respectively $M^{\x{x}}_{0,i}$ the number of bits before the binary point  of the feature layer) which is obtained by computing the ufp defined in Equation~\eqref{ufp} of the corresponding floating-point coefficient. Finally, the value of $M^{\x{u}}_{k,i}$ is obtained through the fixed-point arithmetic (addition and multiplication) in Section~\ref{s3}. 
In Equation~\eqref{c1} of Figure~\ref{cst} (respectively \eqref{c2} and \eqref{c3}), the length $M^{\x{x}}_{k,i}+L^{\x{x}}_{k,i}$ (respectively $M^{\x{u}}_{k,i}+L^{\x{u}}_{k,i}$ and $M^{\x{w}}_{k,i,j}+L^{\x{w}}_{k,i,j}$) of the fixed-point number $\x{x}$ (respectively $\x{u}$ and $\x{w}$) must be less than or equal to $T-1$ to ensure the data type required. We use $T-1$ in these three constraints because we keep one bit for the sign. 
Equation (\ref{c11}) is about the multiplication. It asserts that the total number of bits of $\x{x}$ and $\x{w}$ is not exceeding the data type $T-1$.
Equations~\eqref{c4}, \eqref{c5} and \eqref{c10} assert that the number of significant bits of the fractional parts cannot be negative.
The boundary condition for the neurons of the output layer is represented in Equation~\eqref{c6}. It gives a lower bound for $L^{\x{u}}_{m,i}$ and then ensures that the error threshold is satisfied for all the neurons of the output layer.
\label{s5}
\begin{figure}[!h]
\begin{tabular}{c}
\hline
\hline
\end{tabular}
\begin{tabular}{c}
\hline
\hspace{15.3cm}                                                       \\      
\end{tabular}
\begin{small}
\begin{equation}
\label{c1}
\begin{aligned}
M^{\x{x}}_{k,i}+L^{\x{x}}_{k,i} \leq T-1,	\hspace{0.5cm} 0\leq k\leq m,\hspace{0.1cm} 0\leq i< n
\end{aligned}
\end{equation}
\end{small}
\begin{small}
\begin{equation}
\label{c2}
\begin{aligned}
M^{\x{u}}_{k,i}+L^{\x{u}}_{k,i} \leq T-1, \hspace{0.5cm} 1\leq k < m+1,\hspace{0.1cm} 0\leq i< n
\end{aligned}
\end{equation}
\end{small}
\begin{small}
\begin{equation}
\label{c3}
\begin{aligned}
M^{\x{w}}_{k,i,j}+L^{\x{w}}_{k,i,j} \leq T-1, \hspace{0.5cm} 0\leq k < m,\hspace{0.1cm} 0\leq i,j< n
\end{aligned}
\end{equation}
\end{small}
\begin{small}
\begin{equation}
\label{c11}
\begin{aligned}
M^{\x{w}}_{k,i,j}+L^{\x{w}}_{k,i,j}+ M^{\x{x}}_{k,j}+L^{\x{x}}_{k,j} \leq T-1, \hspace{0.5cm} 0\leq k < m,\hspace{0.1cm} 0\leq i,j< n
\end{aligned}
\end{equation}
\end{small}
\begin{small}
\begin{equation}
\label{c4}
\begin{aligned}
L^{\x{x}}_{k,i} \geq 0, \hspace{0.5cm} 0\leq k\leq m,\hspace{0.1cm} 0\leq i< n
\end{aligned}
\end{equation}
\end{small}
\begin{small}
\begin{equation}
\label{c5}
L^{\x{u}}_{k,i} \geq 0, \hspace{0.5cm} 1\leq k < m+1,\hspace{0.1cm} 0\leq i< n
\end{equation}
\end{small}
\begin{small}
\begin{small}
\begin{equation}
\label{c10}
\begin{aligned}
L^{\x{w}}_{k,i,j} \geq 0, \hspace{0.5cm} 0\leq k< m,\hspace{0.1cm} 0\leq i,j< n
\end{aligned}
\end{equation}
\end{small}
\begin{equation}
\label{c6}
L^{\x{u}}_{m,i} \geq |\text{ufp}(|Threshold|)|, \hspace{0.3cm} 0\leq i< n,\hspace{0.1cm} m:\hspace{0.01cm}\textnormal{last}\hspace{0.1cm} \textnormal{layer }\hspace{0.03cm} \textnormal{of }\hspace{0.05cm} \textnormal{NN}
\end{equation}
\end{small}
 \begin{small}
\begin{equation}
\label{c7}
\forall j \, : \,L^{\x{u}}_{k,i} \leq L^{\x{x}}_{k-1,j}, \hspace{0.123cm}  \hspace{0.1cm} 1\leq k < m+1,\hspace{0.1cm} 0\leq i,j< n
\end{equation}
\end{small}
\begin{small}
\begin{equation}
\label{c8}
L^{\x{x}}_{k,i} \leq L^{\x{u}}_{k,i}, \hspace{0.3cm} 1\leq k < m+1,\hspace{0.1cm} 0\leq i< n
\end{equation}
\end{small}
\begin{footnotesize}
\begin{multline}
\label{c9}
L^{\x{u}}_{k,i}\times (\text{ufp}(n)+1)+ \sum_{j=0}^{n-1}(L^{\x{x}}_{k-1,j}+L^{\x{w}}_{k-1,i,j})\geq \sum_{j=0}^{n-1}(M^{\x{x}}_{k-1,j}+M^{\x{w}}_{k-1,i,j})-\text{ufp}(|Threshold|)-1,\cr \hspace{0.1cm} 1\leq k < m+1,\hspace{0.1cm} 0\leq i,j< n
\end{multline}
\end{footnotesize}
\begin{tabular}{c}
\hline
 \hspace{15.3cm}                                            
\end{tabular}
\caption{\centering Constraints generated for the formats optimization of each neuron of the NN.}
\label{cst}
\end{figure}

In Figure~\ref{cst}, Equation~\eqref{c7} represents the  constraint where the propagation is done in a forward way, and Equation \eqref{c8} represents the  constraint where the propagation is done in a backward way.
These constraints bound the length of the fractional parts in the worst case.
The constraint of Equation~\eqref{c7} aims at giving an upper bound of $L^{\x{u}}_{k,i}$. 
It ensures that $L^{\x{u}}_{k,i}$ of the output of the neuron $i$ of the layer $k$ is less than (or equal to) all its inputs $L^{\x{x}}_{k-1,j}$, $0 \leq j <n$. The constraint of Equation~\eqref{c8} gives an upper bound for $L^{\x{x}}_{k,i}$ of the input $\x{x}$. This constraint ensures that the number of significant bits after the binary point $L^{\x{x}}_{k,i}$ of the input of the neuron $i$ of the layer $k+1$ is equal to (or less than) $L^{\x{u}}_{k,i}$ of the neuron $i$ of the previous layer $k$.
The constraint of Equation~\eqref{c9} in Figure~\ref{cst} aims at bounding $L^{\x{u}}_{k,i}$ of the output of the neuron $i$ for the layer $k$ and $L^{\x{w}}_{k-1,i,j}$ of the coefficients of matrix $\x{W}_{k-1}$. This constraint corresponds to the error done during the computation of the affine transformation in Equation~(\ref{e8}). The Equation~\eqref{c9} is obtained by the linearization of Equation~(\ref{e10}) of Proposition~\ref{p3}, in other words, we have to compute the ufp of the error.
The ufp of the error, written $\text{ufp}(\theta_{k,i})$, is computed as follow\\ \\
Using Equation~\eqref{e24} of Proposition~\ref{p3}, we have\\
\begin{small}
\begin{center}
$\textnormal{ufp}(\theta_{k,i})\leq \textnormal{ufp}(\displaystyle \sum_{j=0}^{n-1}(2^{M^{\hat{x}}_{k-1,j}-L^{\hat{w}}_{k-1,i,j}}+2^{M^{\hat{w}}_{k-1,i,j}-L^{\hat{x}}_{k-1,j}})+ n\times2^{-L^{\hat{u}}_{k,i}}+2^{-L^{\hat{u}}_{k,i}+1})$,
\end{center}
\end{small}
then we obtain
\begin{small}
\begin{center}
$\textnormal{ufp}(\theta_{k,i})\leq \displaystyle \sum_{j=0}^{n-1}({M^{\hat{x}}_{k-1,j}-L^{\hat{w}}_{k-1,i,j}}+{M^{\hat{w}}_{k-1,i,j}-L^{\hat{x}}_{k-1,j}})-L^{\hat{u}}_{k,i}\times({\textnormal{ufp}(n)+1})+1$.
\end{center}
\end{small}
We notice that 
\begin{small}
$ \textnormal{ufp}(\theta_{k,i})\leq \textnormal{ufp}(|Threshold|) \leq 0 $ because the error is between $0$ and $1$.\\
\end{small}
Finally,
\begin{small}
$L^{\x{u}}_{k,i}\times (\text{ufp}(n)+1)+$ $\displaystyle \sum_{j=0}^{n-1}(L^{\x{x}}_{k-1,j}+L^{\x{w}}_{k-1,i,j})\geq \displaystyle \sum_{j=0}^{n-1}(M^{\x{x}}_{k-1,j}+M^{\x{w}}_{k-1,i,j})-\text{ufp}(|Threshold|)-1.$  \hfill \qed   
\end{small}

All the constraints defined in Figure~\ref{cst} are linear with integer variables. The optimal solution is found by solving them by linear programming. This solution gives the minimal 
number of bits for the fractional part required for each neuron $\x{u}_{k,i}$ of each layer taking into account the data type $T$ and the error threshold tolerated by the user in one hand, and on the other hand the minimal number of bits of the fractional part required for each coefficient $\x{w}_{k-1,i,j}$. 
\section{\textbf{Implementation}}
\label{s9}
In this section, we present our tool.
Our approach which is computing the optimal formats $<~M^{\x{u}}_{k,i}, L^{\x{u}}_{k,i}>$ for each neuron $\x{u}_{k,i}$ of each layer for a given NN, satisfying an error threshold between $0$ and $1$ and a data type $T$ given by the user is evaluated through this~tool.

Our tool is a fixed-point code synthesis tool. It synthesizes a C code for a given NN. This code contains arithmetic operations and activation functions, which use the fixed-point arithmetic (integer arithmetic) only. In this article, we present only the $\x{ReLU}$ and $\x{Linear}$ activation functions (defined in Equation~\eqref{relu} and \eqref{linear} respectively) but we can also deal with  $\x{Sigmoid}$ and $\x{Tanh}$ activation functions in our current implementation. They are not shown but they are available in our framework. We have chosen to approximate them through piecewise linear approximation \cite{activ} using fixed-point arithmetic. We compute the corresponding error like in $\x{ReLU}$ and  $\x{Linear}$, then we generate the corresponding constraints. 

A description of our tool is given in Figure~\ref{fig3}. 
\begin{figure}[!h]
\begin{center}
\includegraphics[scale=0.23]{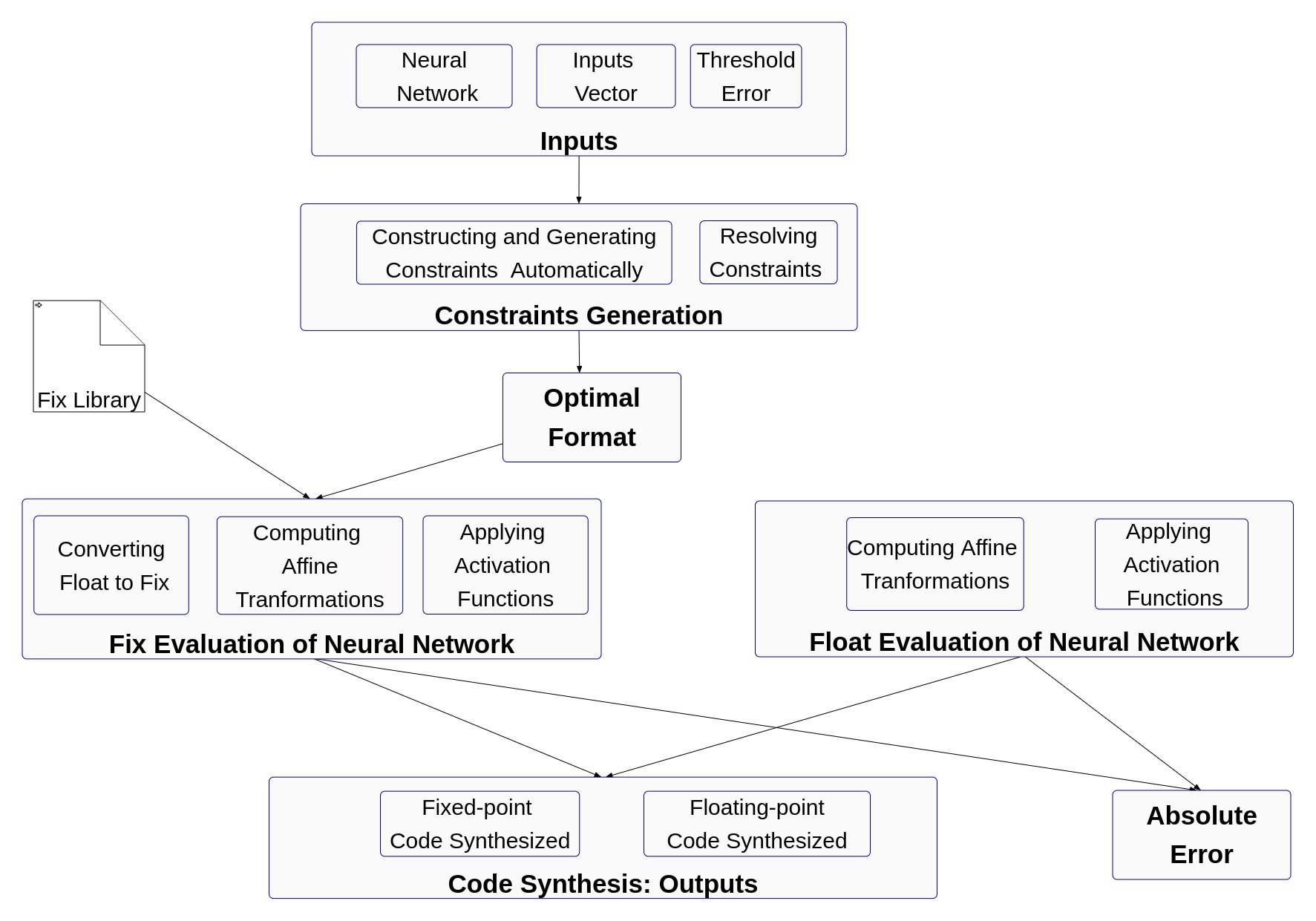}
\caption{Our tool features.}
\label{fig3}
\end{center}
\end{figure}
It takes a floating-point NN working at some precision, input vectors and a threshold error chosen by the user.
The user also has the possibility to choose the data type $T\in\{8, 16, 32\}$ bits wanted for the  code synthesis.
First, we do a preliminary analysis of the NN through many input vectors in order to determine the range of the outputs of the neurons for each layer in the worst case. We compute also the most significant bit of each neuron of each layer in the worst case which gives us the range of the inputs of the NN $[x_{min},x_{max}]$ for which we certify that our synthesized code is valid for any input belong this range respecting the required data type and the threshold.
Our tool generates automatically the constraints mentioned in Figure~\ref{cst} in Section \ref{s5} and solves them using the linprog function of scipy library \cite{welke2020ml2r} in Python. Then, the optimal formats for each output neuron, input neuron and coefficients of biases and matrices of weights are obtained. 
The optimal formats are used for the conversion from the floating-point into fixed-point numbers for all the coefficients
  inside each neuron. To make an evaluation of the NN in fixed-point arithmetic i.e computing the function in Equation~(\ref{e8}), a fixed-point library is needed. Our library contains mainly the following functions:  fixed-point addition and multiplication, shifts, fixed-point activation functions $\x{Tanh}$, $\x{Sigmoid}$, $\x{ReLU}$ and $\x{Linear}$. 
  The conversion of a fixed-point number to a floating-point number and the conversion of a floating-point number to a fixed-point number is also available in this library. 
The last step consists of the fixed-point code synthesis. We also synthesize a floating-point code to make a comparison with the fixed-point synthesized code.
We show experiments and results of some NNs in Section \ref{s6}, and we compare them with floating-point results in terms of memory space (bits saved) and accuracy.
\section{\textbf{Experiments}}
\label{s6}
In this section, we show some experimental results done on seven trained NNs. Four of them are interpolators which compute mathematical functions and the three others are classifiers. These NNs are described in Table \ref{desc}. The first column gives the name of the NNs. The second column gives the number of layers and the third one shows the number of neurons. The number of connections between the neurons is shown in the fourth column and the number of generated constraints for each NN by our approach is given in the last column.
\begin{table}[h]
\centering
\caption{\centering \textnormal{Description of our experimental NNs.}}
\label{desc}
\begin{tabular}{|c|c|c|c|c|}
\hline
\diagbox{NN}{Desc.} & 
Layers   & Neurons
    & Connections & Constraints    \\ \hline
\textit{\textbf{hyper}} & $4$  & $48$ & $576$ & 1980  \\ \hline
\textit{\textbf{bumps}} & $2$  & $60$ & $1800$ & 5010  \\ \hline
% \textit{\textbf{AFun}}&4  & 40 &400 & 1430  \\ \hline
\textit{\textbf{CosFun}} &4   &40 &400  & 1430  \\ \hline
\textit{\textbf{Iris}} &3   &33 &363 & 1243 \\ \hline
\textit{\textbf{Wine}} &2   &52 &1352 & 3822  \\ \hline
\textit{\textbf{Cancer}} &3   &150 &7500 & 21250  \\ \hline
\textit{\textbf{AFun}} &2   &200 &10000 & 51700  \\ \hline
\end{tabular}

\hspace{0.3cm}
\end{table}

The NN \textit{hyper} in Table \ref{desc} is an interpolator network computing the hyperbolic sine of the point $(x, \, y)$. It is made of four layers, $48$ fully connected neurons ($576$ connections.) The number of constraints generated for this NN in order to compute the optimal formats is $1980$. The NN \textit{bumps} is an interpolator network computing the $bump$ function. The affine function $f(x)=4\times x+\frac{3}{2}$ is computed by the \textit{AFun} NN (interpolator) and the function $f(x,y)=x\times cos(y)$ is computed by the \textit{CosFun} NN (interpolator.)
The classifier \textit{Iris} is a NN which classifies the \textit{Iris} plant \cite{swain2012approach} into three classes: \textit{Iris-Setosa}, \textit{Iris-Versicolour} and \textit{Iris-Virginica}. It takes four numerical attributes as input (sepal length in cm, sepal width in cm, petal length in cm and petal width in cm.) The NN \textit{Wine} is also a classifier. It classifies wine into three classes \cite{aeberhard1992classification} through thirteen numerical attributes (alcohol, malic acid, ash, etc.) The last one is the \textit{Cancer} NN  which classifies the cancer into two categories (malignant and benign) through thirty numerical attributes \cite{wolberg1992breast} as input.
 These NNs originally work in IEEE754 single precision. We have transformed them into fixed-point NNs satisfying a threshold error and a data type $T$ (the size of the fixed-point numbers in the synthesized code) set by the user. Then we apply the $\x{ReLU}$ activation function defined in Equation~\eqref{relu} or the $\x{Linear}$ activation function defined in Equation \eqref{linear} (or $\x{Sigmoid}$ and~$\x{Tanh}$ also.) 
\subsection{Accuracy and Error Threshold}
The first part of experiments is for accuracy and error threshold. It concerns Table~\ref{cmp1}, Figure~\ref{cosfun_3d} and Figure~\ref{histo} and shows if the concerned NNs satisfy the error threshold set by the user using the data type $T$. If the NN is an interpolator, it means that the output of the mathematical function $f$ has an error less than or equal to the threshold. If the NN is a classifier, it means that the error of classification of the NN is less than or equal to $(threshold\times 100)\%$.

 The symbol  $\times$ in Table \ref{cmp1} refers to the infeasability of the solution when the linear programming \cite{welke2020ml2r} fails to find a solution or when we cannot satisfy the threshold using the data type $T$. The symbol $\surd$ means that our linear solver has found a solution to the system of constraints (Section~\ref{s5}.) 
Each line of Table~\ref{cmp1} corresponds to a given NN in some precision and the columns correspond to the multiple values of the error thresholds. 
For example, in the first line of the first column, the NN $hyper\_32$ requires a data type $T=32$~bits and satisfies all the values of threshold (till $10^{-6}\approx 2^{-20}$.) In the fifth column, $2^{-4}$ means that we require at least four significant bits in the fractional part of the worst output of the last layer of the NNs. 
The fixed-point NNs \textit{bumps\_32, AFun\_32, CosFun\_{32}, Iris\_{32} \textnormal{and} Cancer\_{32} } fulfill the threshold value $2^{-10}$ which corresponds to ten accurate bits in the fractional part of the worst output. Beyond this value, the linear programming \cite{welke2020ml2r} does not find a solution using a data type on $32$~bits for these NNs. 
Using the data type $T=16$~bits, all the NNs except \textit{AFun\_16} have an error less than or equal to $2^{-4}$ and ensure the correctness at least of four bits after the binary point of the worst output in the last layer.
Only the NNs \textit{Wine\_8} and \textit{Cancer\_8} can ensure one significant bit after the binary point using data type on $8$ bits. In the other NNs, only the integer part is correct.

The results can vary depending on several parameters: the input vector, the coefficients of $W$ and $b$, the activation functions, the error threshold and the data type $T$. Generally, when the coefficients are between $-1$ and $1$, the results are more accurate because their ufp (Equation\eqref{ufp}) are negative and we can go far after the binary point. The infeasability of the solutions depends also on the size of the data types $T$, for example if we have a small data type $T$ and a consequent number of bits before the binary point in the coefficients (large value of $M$), we cannot have enough bits after the binary point to satisfy the small error thresholds.

Figure \ref{cosfun_3d} represents the fixed-point outputs of the interpolator \textit{CosFun\_32} (lines) and the floating-point outputs (shape) for multiple inputs. All the fixed-point  outputs must respect the threshold $2^{-10}$ (ten significant bits in the fractional part) and must use the data type $T=32$ bits in this case.
We can see that the two curves are close. This means that the new NN (fixed-point) has the same behavior and answer comparing to the original NN. The result is correct for this NN for any inputs $x\in[-4,4]$ and $y\in[-4,4]$.
\begin{table*}
\caption{\centering \textnormal{Comparison between the multiple values of error thresholds set by the user and our tool experimental errors using data types on $32$, $16$ and $8$ bits for the fixed-point synthesized code.}\label{cmp1}}
 \begin{tabularx}{\linewidth}{@{}c|@{}c|@{}c *8{>{\centering\arraybackslash}X}@{}@{}}
\toprule
\hline
\textbf{Data Types}& \diagbox[width=2.1cm,height=3.25\line]{\textbf{NN}}{\textbf{Threshold}}     & \textbf{$10^{0}=\hspace{0.0cm}2^{0}$} & \textbf{$0.5=\hspace{0.5cm}2^{-1}$} & \textbf{$10^{-1}\approx \hspace{0.05cm}2^{-4}$}& \textbf{$10^{-2}\hspace{0.05cm}\approx2^{-7}$} & \textbf{$10^{-3}\approx2^{-10}$}& \textbf{$10^{-4}\approx2^{-14}$}& \textbf{$10^{-5}\approx2^{-17}$} & \textbf{$10^{-6}\approx2^{-20}$}   \\ 
\midrule
\hline
\multirow{6}*{\textbf{32 bits}}&\textit{\textbf{hyper\_$32$}}& $\surd$&$\surd$ &$\surd$ &$\surd$ &$\surd$ & $\surd$&$\surd$ & $\surd$\\
&\textit{\textbf{bumps\_$32$}}& $\surd$& $\surd$&$\surd$ & $\surd$&$\surd$ &$\times$ &$\times$ &$\times$ \\
&\textit{\textbf{AFun\_$32$}} & $\surd$& $\surd$&$\surd$ & $\surd$&$\surd$ &$\times$ &$\times$ &$\times$ \\
&\textit{\textbf{CosFun\_$32$}}&$\surd$&$\surd$ &$\surd$ &$\surd$ &$\surd$ & $\times$&$\times$ & $\times$ \\
&\textit{\textbf{Iris\_$32$}}&$\surd$& $\surd$&$\surd$ & $\surd$&$\surd$ &$\times$ &$\times$ &$\times$ \\
&\textit{\textbf{Wine\_$32$}}& $\surd$& $\surd$&$\surd$ & $\surd$&$\surd$ &$\surd$ &$\times$ &$\times$\\
&\textit{\textbf{Cancer\_$32$}}& $\surd$& $\surd$&$\surd$ & $\surd$&$\surd$ &$\times$ &$\times$ &$\times$\\
%\addlinespace
\hline
\hline
\multirow{6}*{\textbf{16 bits}}&\textit{\textbf{hyper\_$16$}} &  $\surd$ &$\surd$& $\surd$& $\times$&$\times$ &$\times$&$\times$&$\times$      \\ 
&\textit{\textbf{bumps\_$16$}}       &  $\surd$ &$\surd$& $\surd$& $\times$&$\times$ &$\times$&$\times$&$\times$      \\ 
&\textit{\textbf{AFun\_$16$}}      &$\surd$ &$\surd$& $\times$& $\times$&$\times$ &$\times$&$\times$&$\times$    \\ 
&\textit{\textbf{CosFun\_$16$}}  &  $\surd$ &$\surd$& $\surd$& $\times$&$\times$ &$\times$&$\times$&$\times$        \\ 
&\textit{\textbf{Iris\_$16$}}      & $\surd$& $\surd$ & $\surd$ & $\times$&$\times$ &$\times$&$\times$&$\times$      \\ 
&\textit{\textbf{Wine\_$16$}}& $\surd$ &$\surd$& $\surd$& $\times$&$\times$ &$\times$&$\times$&$\times$ \\
&\textit{\textbf{Cancer\_$16$}}& $\surd$ &$\surd$& $\surd$& $\times$&$\times$ &$\times$&$\times$&$\times$ \\
\hline \hline
\multirow{6}*{\textbf{8 bits}}&\textit{\textbf{hyper\_$8$}}   &  $\times$ & $\times$& $\times$&$\times$ &$\times$&$\times$&$\times$  &$\times$          \\ 
&\textit{\textbf{bumps\_$8$}}  & $\surd$ & $\times$& $\times$&$\times$ &$\times$&$\times$&$\times$  &$\times$       \\ 
&\textit{\textbf{AFun\_$8$}}       & $\times$ & $\times$& $\times$&$\times$ &$\times$&$\times$&$\times$  &$\times$      \\ 
&\textit{\textbf{CosFun\_$8$}}         & $\times$ & $\times$& $\times$&$\times$ &$\times$&$\times$&$\times$  &$\times$    \\ 
&\textit{\textbf{Iris\_$8$}}        &  $\surd$ &  $\times$& $\times$&$\times$ &$\times$&$\times$&$\times$  &$\times$    \\ 
&\textit{\textbf{Wine\_$8$}}&$\surd$ &$\surd$ &  $\times$&$\times$ &$\times$&$\times$&$\times$  &$\times$ \\
&\textit{\textbf{Cancer\_$8$}}&$\surd$ &$\surd$ &  $\times$&$\times$ &$\times$&$\times$&$\times$  &$\times$ \\
\bottomrule
\end{tabularx}
\end{table*}

\begin{figure}[!h]
   \begin{minipage}[c]{.4\linewidth}
      \includegraphics[width=7.5cm]{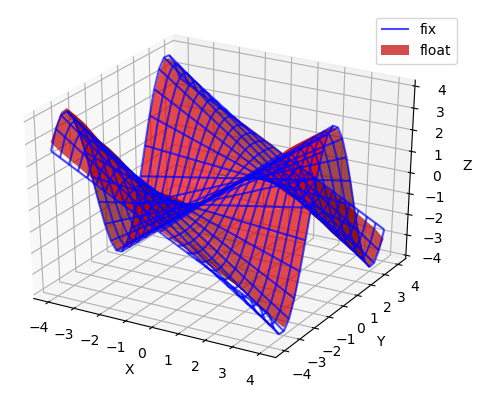}
      \caption{\centering Fixed-point outputs vs floating-point outputs of the \textit{CosFun\_32} NN for multiple inputs using a data type on $32$ bits and an error threshold $\approx 2^{-10}$.}
\label{cosfun_3d}
   \end{minipage} \hfill
   \begin{minipage}[c]{.46\linewidth}
      \includegraphics[width=6cm]{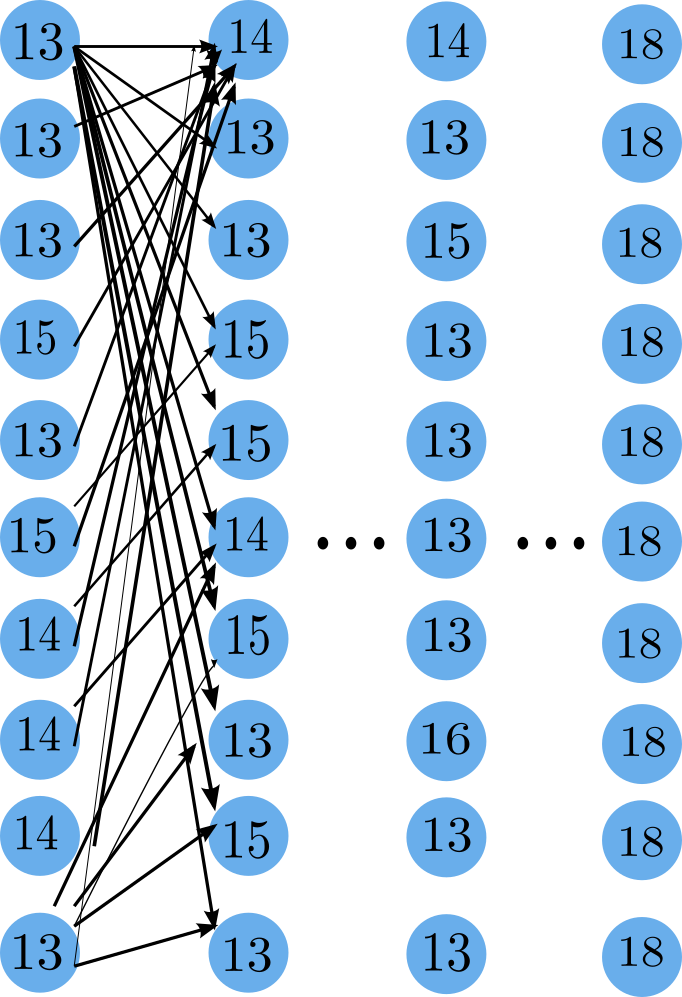}
      \caption{\centering Number of bits of each neuron of the \textit{CosFun\_32} NN after formats optimization for a threshold $\approx 2^{-10}$.}
\label{cosfun}
   \end{minipage}
\end{figure}
Figure~\ref{histo} shows the results of the fixed-point classifications for the NNs \textit{Iris\_32} (right) and \textit{Wine\_32} (left) using a data type $T=32$ bits and a threshold value $2^{-7}$ for multiple input vectors ($=8$).
For example, the output corresponding to the input vectors $1$ and $2$ of the \textit{Iris\_32} NN is \textit{Iris-Versicolour}, for the input vectors $3$, $5$ and $6$ \textit{Iris-Virginica} and \textit{Iris-Setosa} for the others. The results are interpreted in the same way for the \textit{Wine\_32} NN. We notice that we have obtained the same classifications with the floating-point NNs using the same input vectors, so our new NN has the same behavior as the initial NN.
\begin{figure*}[!h]
\begin{center}
\includegraphics[scale=0.27]{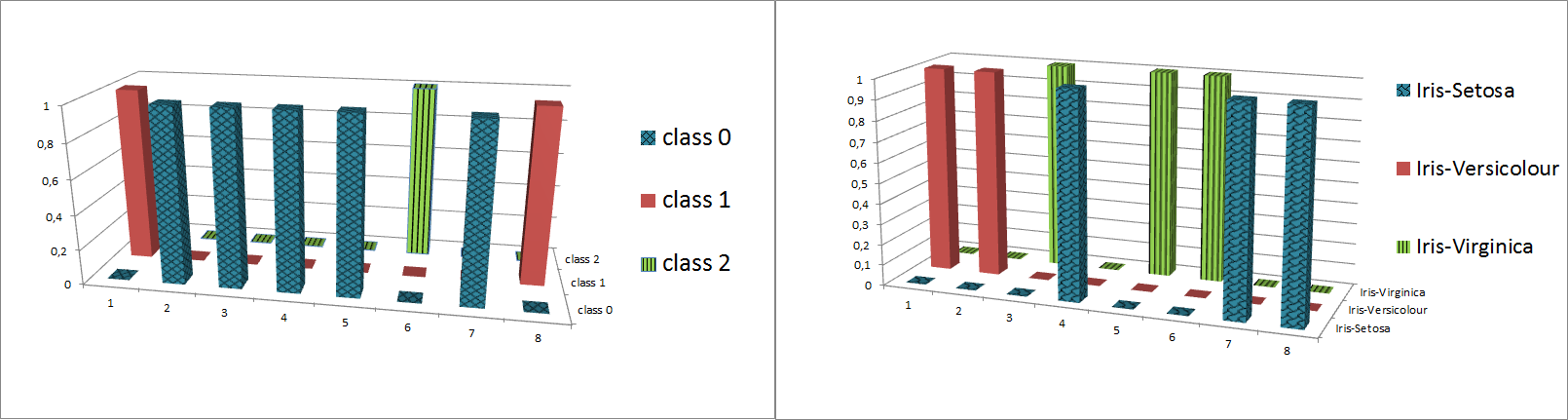}
\caption{\centering Results of the fixed-point classification of the NNs \textit{Iris\_32} (right) and \textit{Wine\_32} (left) respecting the threshold value $2^{-7}$ and~the~data~type~$T=32$~bits.}
\label{histo}
\end{center}
\end{figure*}
\subsection{Bits/Bytes Saved}
The second part of experiments concerns Figure~\ref{cosfun}, Figure~\ref{numberbit} and Table~\ref{gain} and aims at showing that our approach saves bytes/bits through the computation of the optimal format for each neuron done in Section~\ref{s5}. At the beginning, all the neurons are represented in $T\in \{8,\, 16,\, 32\}$~bits and after the optimization step, we reduce consequently the number of bits for each neuron while respecting the threshold set by the user.  

Figure \ref{cosfun} shows the total number of bits of all the neurons for each layer of the \textit{CosFun\_32} NN after the optimization of the formats $<M,$ $L>$ in order to satisfy the threshold $2^{-10}$ and the data type $T=32$ bits for the fixed-point synthesized code in this case. In the synthesized C code, the data types will not change through this optimization because they are defined at the beginning of the program statically, but it is interesting to use these results in FPGA \cite{bevcvavr2005fixed,GSS19} for example. 
We notice that the size of all the neurons was $32$ bits at the beginning and after our tool optimization (resolving the constraints of Section \ref{s5}), we need only $18$ bits for the neurons of the output layer to satisfy the error threshold. We win $14$ bits which represents a gain of $43.75\%$.

\begin{figure}[!h]
\centerline{\includegraphics[scale=0.65]{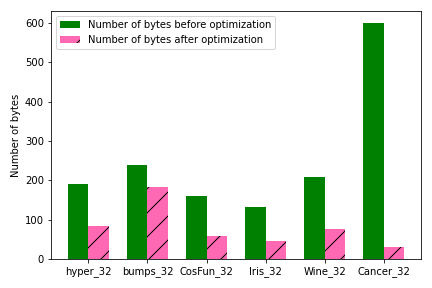}}
\caption{\centering Initial vs saved bytes for our experimental NNs on $32$ bits respecting the threshold $2^{-7}$.}
\label{numberbit}
\end{figure}

In the initial NNs, all the neurons were represented on $32$ bits ($4$ bytes) but after the formats optimization (Section~\ref{s5}) through linear programming \cite{welke2020ml2r}, we save many bytes for each neuron of each layer and the size of the NN becomes considerably small.
It is useful in the case when we use FPGA \cite{bevcvavr2005fixed,GSS19}, for example. Figure~\ref{numberbit} shows the size of each NN (bytes) before and after optimization, and the Table \ref{gain} demonstrates the percentage of gain for each NN. For example, in the NN \textit{Cancer\_32}, we reduce for  $18\times$ the number of bytes comparing to the initial NN (we earn up $94,66\%$.) Our approach saves bits and takes in consideration the threshold  error set by the user (in this example it is $2^{-7}$.)
\begin{table}[h]
\centering
\caption{\centering \textnormal{Gain of bytes of the experimental NNs after formats optimization for a threshold $\approx 2^{-7}$ and a data type $T=32$ bits.}}
\label{gain}
\begin{small}
\begin{tabular}{|c|c|c|c|c|}
\hline
\diagbox{NN}{Desc.} & 
Size before opt. & Size after opt. & Bytes saved
    & Gain (\%)   \\ \hline
\textit{\textbf{hyper\_32}} & $192$ & $84$ & $108$ & \textbf{$56,25$}   \\ \hline
\textit{\textbf{bumps\_32}} & $240$ & $183$ & $57$ & \textbf{$23.75$}  \\ \hline
\textit{\textbf{CosFun\_32}}&$160$ & $59$ & $101$ &\textbf{$63.12$}  \\ \hline
\textit{\textbf{Iris\_32}} &$132$ &$47$  &$85$ &\textbf{$64.39$}  \\ \hline
\textit{\textbf{Wine\_32}} &$208$ &$77$  &$131$ &\textbf{$62.98$}  \\ \hline
\textit{\textbf{Cancer\_32}} &$600$ &$32$  &$568$ &\textbf{$94.66$}  \\ \hline
\end{tabular}
\hspace{0.3cm}
\end{small}
\end{table}
\subsection{Conclusion}
These experimental results show the efficiency of our approach in terms of accuracy and bits saved (memory space.)
As we can see in Figure~\ref{fig_code}, the synthesized code contains only assignments, elementary operations \textit{$(+,\, \times,\, >>,\, <<)$} and conditions used in the activation functions. The execution time corresponding to the synthesized code for the experimental NNs of the Table \ref{desc} is in only few milliseconds (ms).
\section{\textbf{Related Work}}
\label{s7}
Recently, a new line of research has emerged on compressing machine learning models~\cite{joseph2020programmable, JinDLTTC19}, using other arithmetics in order to run NNs on devices with small memories, integer CPUs~\cite{GSS19,HanZZWL19} and optimizing data types and computations error~\cite{ioualalen2019neural, LauterV20}.

In this section, we give an overview of some recent work. We present the multiple tools and frameworks (SEEDOT \cite{GSS19}, DEEPSZ~\cite{JinDLTTC19}, Condensa~\cite{joseph2020programmable}) more or less related to our approach. There is no approach comparable with our method because none of them respects a threshold error set by the user in order to synthesize a C code using only integers for a given trained NN without modifying its behavior.
We can cite also FxpNet \cite{ChenHZX17} and Fix-Net \cite{enderich2019fix} which train neural networks using fixed point arithmetic (low bit-width arithmetic) in both forward pass and backward pass. The articles  \cite{enderich2019fix,LinTA16} are about quantization which aims to reduce the complexity of DNNs and facilitate potential deployment on embedded hardware.
There is also another line of research who has emerged recently on understanding safety and robustness of NNs \cite{AI2,DeepPoly,Sherlock,tran2020nnv,katz2017reluplex}. We can mention the frameworks Sherlock \cite{Sherlock}, AI$^{2}$~\cite{AI2}, DeepPoly~\cite{DeepPoly} and NNV~\cite{tran2020nnv}.

The SEEDOT framework \cite{GSS19} synthesizes a fixed-point code for  machine learning (ML) inference algorithms that can run on constrained hardware. This tool presents a compiling strategy that reduces the search space for some key parameters, especially scale parameters for the fixed-point numbers representation used in the synthesized fixed-point code. Some operations are implemented (multiplication, addition, exponential, argmax, etc.) in this approach.
Both of SEEDOT and our tool generate fixed-point code, but our tool  fullfills a threshold and a data type required by the user. SEEDOT finds a scale for the fixed-point representation number and our tool solves linear constraints for finding the optimal format for each neuron.
The main goal of this compiler is the optimization of the fixed-point arithmetic numbers and operations for an FPGA and micro-controllers.

The key idea of \cite{ioualalen2019neural} is to reduce the sizes of data types used to compute inside each neuron of the network (one type per neuron) working in IEEE754 floating-point arithmetic \cite{IEEE,martel2017floating}. The new NN with smaller data types behaves almost like the original NN with a percentage error tolerated.
This approach generates constraints and does a forward and a backward analysis to bound each data type. 
Our tool has a common step with this approach, which is the generation of constraints for finding the optimal format for each neuron (fixed-point arithmetic) for us and the optimal size (floating-point arithmetic) for each neuron for this method.

In\cite{HanZZWL19}, a new data type called Float-Fix is proposed. This new data type is a trade-off between the fixed-point arithmetic \cite{lopez2014implementation,najahi2014synthesis,bevcvavr2005fixed} and the floating-point arithmetic \cite{IEEE,martel2017floating}. This approach analyzes the data distribution and data precision in NNs then applies this new data type in order to fulfill the requirements. The elementary operations are designed for Float-Fix data type and tested in the hardware. The common step with our approach is the analysis of the NN and the range of its output in order to find the optimal format using the fixed-point arithmetic for us and the the optimal precision for this method using Float-Fix data type. Our approach takes a threshold error not to exceed but this approach does not.

DEEPSZ \cite{JinDLTTC19} is a lossy compression framework. It compresses sparse weights in deep NNs using the floating-point arithmetic. DEEPSZ involves four key steps: network pruning, error bound assessment, optimization for error bound configuration and compressed model generation. A threshold is set for each fully connected layer, then the weights of this layer are pruned. Every weight below this threshold is removed. This framework determines the best-fit error bound for each layer in the network, maximizing the overall compression ratio with user acceptable loss of inference accuracy.

The idea presented in \cite{joseph2020programmable} is about using weight pruning and quantization for the compression of deep NNs \cite{sun2018testing}. The model size and the inference time are reduced without appreciable loss in accuracy. The tool introduced is Condensa where the reducing memory footprint is by zeroing out individual weights and reducing inference latency is by pruning 2-D blocks of non-zero weights for a language translation network (Transformer).

A framework for semi-automatic floating-point error analysis for the inference phase of deep learning is presented in \cite{LauterV20}. It transforms a NN into a C++ code in order to analyze the network need for precision. The affine and interval arithmetics are used in order to compute the relative and absolute errors bounds for deep NN \cite{sun2018testing}.
This article gives some theoretical results which are shown for bounding and interpreting the impact of rounding errors due to the precision choice for inference in generic deep NN \cite{sun2018testing}.
\section{\textbf{Conclusion \& Future Work}}
\label{s8}
In this article, we introduced a new approach to synthesize a fixed-point code for NNs using the fixed-point arithmetic and to tune the formats of the computations and conversions done inside the neurons of the network. This method ensures that the new fixed-point NN still answers correctly compared to the original network based on IEEE754 floating-point arithmetic \cite{IEEE}. This approach ensures the non overflow (sufficient bits for the integer part) of the fixed-point numbers in one hand and the other hand, it respects the threshold required by the user (sufficient bits in the fractional part.) It takes in consideration the propagation of the round off errors and the error of inputs through a set of linear constraints among integers, which can be solved by linear programming \cite{welke2020ml2r}.
Experimental results show the efficiency of our approach in terms of accuracy, errors of computations and bits saved.
The limit of the current implementation is the large number of constraints. We use linprog in Python \cite{welke2020ml2r} to solve them but this method does not support a high number of constraints, this is why our experimental NNs are small.

A first perspective is about using another solver to solve our constraints (Z3 for example \cite{Z3}) which deals with a large number of constraints.
A second perspective is to make a comparison study between Z3 \cite{Z3} and linprog \cite{welke2020ml2r} in term of time execution and memory consumption. 
A third perspective is to test our method on larger, real-size industrial neural networks. We believe that our method will scale up as long as the linear programming solver will scale up. If this is not enough, a solution would be to assign the same format to a group of neurons in order to reduce the number of equations and variables in the constraints system. 
A last perspective is to consider the other NN architectures like convolutional NNs \cite{albawi2017understanding,AI2,abraham2005artificial}.

 \bibliographystyle{splncs04}
 \bibliography{aisca22}

\begin{thebibliography}{10}
\providecommand{\url}[1]{\texttt{#1}}
\providecommand{\urlprefix}{URL }
\providecommand{\doi}[1]{https://doi.org/#1}

\bibitem{abraham2005artificial}
Abraham, A.: Artificial neural networks. Handbook of measuring system design
  (2005)

\bibitem{aeberhard1992classification}
Aeberhard, S., Coomans, D., de~Vel, O.: The classification performance of
  {RDA}. Dept. of Computer Science and Dept. of Mathematics and Statistics,
  James Cook University of North Queensland, Tech. Rep pp. 92--01 (1992)

\bibitem{albawi2017understanding}
Albawi, S., Mohammed, T.A., Al-Zawi, S.: Understanding of a convolutional
  neural network. In: 2017 International Conference on Engineering and
  Technology (ICET). pp.~1--6. IEEE (2017)

\bibitem{bevcvavr2005fixed}
Be{\v{c}}v{\'a}{\v{r}}, M., {\v{S}}tukjunger, P.: Fixed-point arithmetic in
  {FPGA}. Acta Polytechnica  \textbf{45}(2) (2005)

\bibitem{catrina2010secure}
Catrina, O., de~Hoogh, S.: Secure multiparty linear programming using
  fixed-point arithmetic. In: Computer Security - {ESORICS} 2010, 15th European
  Symposium on Research in Computer Security. vol.~6345, pp. 134--150. Springer
  (2010)

\bibitem{activ}
{\c{C}}etin, O., Temurta{\c{s}}, F., G{\"u}lg{\"o}n{\"u}l, {\c{S}}.: An
  application of multilayer neural network on hepatitis disease diagnosis using
  approximations of sigmoid activation function. Dicle Medical Journal/Dicle
  Tip Dergisi  \textbf{42}(2) (2015)

\bibitem{ChenHZX17}
Chen, X., Hu, X., Zhou, H., Xu, N.: Fxpnet: Training a deep convolutional
  neural network in fixed-point representation. In: 2017 International Joint
  Conference on Neural Networks, {IJCNN} 2017, Anchorage, AK, USA, May 14-19,
  2017. pp. 2494--2501. {IEEE} (2017)

\bibitem{Sherlock}
Dutta, S., Jha, S., Sankaranarayanan, S., Tiwari, A.: Output range analysis for
  deep feedforward neural networks. In: {NASA} Formal Methods - 10th
  International Symposium, {NFM}. vol. 10811, pp. 121--138. Springer (2018)

\bibitem{enderich2019fix}
Enderich, L., Timm, F., Rosenbaum, L., Burgard, W.: Fix-net: pure fixed-point
  representation of deep neural networks. ICLR  (2019)

\bibitem{AI2}
Gehr, T., Mirman, M., Drachsler{-}Cohen, D., Tsankov, P., Chaudhuri, S.,
  Vechev, M.T.: {AI2:} safety and robustness certification of neural networks
  with abstract interpretation. In: 2018 {IEEE} Symposium on Security and
  Privacy. pp. 3--18. {IEEE} Computer Society (2018)

\bibitem{GSS19}
Gopinath, S., Ghanathe, N., Seshadri, V., Sharma, R.: Compiling {KB}-sized
  machine learning models to tiny {IoT} devices. In: Programming Language
  Design and Implementation, {PLDI} 2019. pp. 79--95. {ACM} (2019)

\bibitem{HanZZWL19}
Han, D., Zhou, S., Zhi, T., Wang, Y., Liu, S.: Float-fix: An efficient and
  hardware-friendly data type for deep neural network. Int. J. Parallel
  Program.  \textbf{47}(3),  345--359 (2019)

\bibitem{IEEE}
{IEEE}: {IEEE} standard for floating-point arithmetic. {IEEE} Std 754-2008 pp.
  1--70 (2008)

\bibitem{ioualalen2019neural}
Ioualalen, A., Martel, M.: Neural network precision tuning. In: Quantitative
  Evaluation of Systems, 16th International Conference, {QEST}. vol. 11785, pp.
  129--143. Springer (2019)

\bibitem{JinDLTTC19}
Jin, S., Di, S., Liang, X., Tian, J., Tao, D., Cappello, F.: Deepsz: {A} novel
  framework to compress deep neural networks by using error-bounded lossy
  compression. In: In International Symposium on High-Performance Parallel and
  Distributed Computing, {HPDC}. pp. 159--170. {ACM} (2019)

\bibitem{joseph2020programmable}
Joseph, V., Gopalakrishnan, G., Muralidharan, S., Garland, M., Garg, A.: A
  programmable approach to neural network compression. {IEEE} Micro
  \textbf{40}(5),  17--25 (2020)

\bibitem{katz2017reluplex}
Katz, G., Barrett, C.W., Dill, D.L., Julian, K., Kochenderfer, M.J.: Reluplex:
  An efficient {SMT} solver for verifying deep neural networks. In: Computer
  Aided Verification, {CAV}. vol. 10426, pp. 97--117. Springer (2017)

\bibitem{LauterV20}
Lauter, C.Q., Volkova, A.: A framework for semi-automatic precision and
  accuracy analysis for fast and rigorous deep learning. In: 27th {IEEE}
  Symposium on Computer Arithmetic, {ARITH} 2020. pp. 103--110. {IEEE} (2020)

\bibitem{LinTA16}
Lin, D.D., Talathi, S.S., Annapureddy, V.S.: Fixed point quantization of deep
  convolutional networks. In: Balcan, M., Weinberger, K.Q. (eds.) Proceedings
  of the 33nd International Conference on Machine Learning, {ICML} 2016, New
  York City, NY, USA, June 19-24, 2016. {JMLR} Workshop and Conference
  Proceedings, vol.~48, pp. 2849--2858. JMLR.org (2016)

\bibitem{lopez2014implementation}
Lopez, B.: Impl{\'{e}}mentation optimale de filtres lin{\'{e}}aires en
  arithm{\'{e}}tique virgule fixe. (Optimal implementation of linear filters in
  fixed-point arithmetic). Ph.D. thesis, Pierre and Marie Curie University,
  Paris, France (2014)

\bibitem{martel2017floating}
Martel, M.: Floating-point format inference in mixed-precision. In: {NASA}
  Formal Methods - 9th International Symposium, {NFM}. vol. 10227, pp. 230--246
  (2017)

\bibitem{Z3}
de~Moura, L.M., Bj{\o}rner, N.: {Z3:} an efficient {SMT} solver. In:
  Ramakrishnan, C.R., Rehof, J. (eds.) Tools and Algorithms for the
  Construction and Analysis of Systems, 14th International Conference, {TACAS}
  2008, Held as Part of the Joint European Conferences on Theory and Practice
  of Software, {ETAPS} 2008, Budapest, Hungary, March 29-April 6, 2008.
  Proceedings. Lecture Notes in Computer Science, vol.~4963, pp. 337--340.
  Springer (2008)

\bibitem{najahi2014synthesis}
Najahi, M.A.: Synthesis of certified programs in fixed-point arithmetic, and
  its application to linear algebra basic blocks. Ph.D. thesis, University of
  Perpignan, France (2014)

\bibitem{sharma2017activation}
Sharma, S., Sharma, S.: Activation functions in neural networks. Towards Data
  Science  \textbf{6}(12),  310--316 (2017)

\bibitem{DeepPoly}
Singh, G., Gehr, T., P{\"{u}}schel, M., Vechev, M.T.: An abstract domain for
  certifying neural networks. Proc. {ACM} Program. Lang., POPL  \textbf{3},
  41:1--41:30 (2019)

\bibitem{sun2018testing}
Sun, Y., Huang, X., Kroening, D., Sharp, J., Hill, M., Ashmore, R.: Testing
  deep neural networks. arXiv preprint arXiv:1803.04792  (2018)

\bibitem{swain2012approach}
Swain, M., Dash, S.K., Dash, S., Mohapatra, A.: An approach for iris plant
  classification using neural network. International Journal on Soft Computing
  \textbf{3}(1), ~79 (2012)

\bibitem{tran2020nnv}
Tran, H., Yang, X., Lopez, D.M., Musau, P., Nguyen, L.V., Xiang, W., Bak, S.,
  Johnson, T.T.: {NNV:} the neural network verification tool for deep neural
  networks and learning-enabled cyber-physical systems. In: Computer Aided
  Verification - 32nd International Conference, {CAV}. vol. 12224, pp. 3--17.
  Springer (2020)

\bibitem{abs-2104-02466}
Urban, C., Min{\'{e}}, A.: A review of formal methods applied to machine
  learning. CoRR  \textbf{abs/2104.02466} (2021)

\bibitem{welke2020ml2r}
Welke, P., Bauckhage, C.: Ml2r coding nuggets: Solving linear programming
  problems. Tech. rep., Technical Report. MLAI, University of Bonn (2020)

\bibitem{wolberg1992breast}
Wolberg, W.H., Street, W.N., Mangasarian, O.L.: Breast cancer wisconsin
  (diagnostic) data set [http://archive.(ics. uci. edu/ml/] (1992)

\bibitem{yates2009fixed}
Yates, R.: Fixed-point arithmetic: An introduction. Digital Signal Labs
  \textbf{81}(83), ~198 (2009)

\end{thebibliography}
\vspace{2cm}

\section*{Authors}
\noindent {\bf H Benmaghnia} received a Master degree in High Performance Computing and Simulations from the University of Perpignan in France. She did a Bachelor of Informatic Systems at the University of Tlemcen in Algeria. Currently, she is pursuing her PhD in Computer Science at the University of Perpignan in LAboratoire de Modélisation Pluridisciplinaire et Simulations (LAMPS). Her research interests include Computer Arithmetic, Precision Tuning, Numerical Accuracy, Neural Networks and Formal Methods.\\

\noindent {\bf M Martel} is a professor in Computer Science
at the University of Perpignan in LAboratoire de Modélisation Pluridisciplinaire et Simulations (LAMPS), France. He is also co-founder and scientific advisor of the start-up Numalis, France. His research interests include Computer Arithmetic, Numerical Accuracy, Abstract Interpretation,  Semantics-based Code Transformations \& Synthesis, Validation of Embedded Systems, Safety of Neural Networks \& Arithmetic Issues, Green \& Frugal Computing, Precision Tuning \& Scientific Data~Compression.\\

\noindent {\bf Y Seladji} received PhD degree in computer science from Ecole Polytechnique, France. Currently, she is associate professor in the University of Tlemcen. Her research interests include Formal Methods, Static Analysis, Embedded Systems.\\

\end{document}